\documentclass[11pt,a4paper]{article}
\usepackage[utf8]{inputenc}
\usepackage[english]{babel}
\usepackage{verbatim}
\usepackage{amsmath,amssymb,array}
\usepackage{amsthm}
\usepackage[algoruled,linesnumbered,noend]{algorithm2e}
\usepackage{tikz}
\usetikzlibrary{arrows,shapes,through}

\newtheorem{Theorem}{Theorem}
\newtheorem{Lemma}[Theorem]{Lemma}
\newtheorem{Proposition}[Theorem]{Proposition}

\newtheorem{Problem}{Problem}
\newtheorem{Example}{Example}

\newcommand{\tyb}{\emph{type b }}
\newcommand{\tya}{\emph{type a }}
\newcommand{\grafo}{\mathcal{G}}

\newcommand{\smAP}{$\langle |\Sigma|,m \rangle$-AP\xspace}
\newcommand{\Rsafe}{\ensuremath{R_{safe}}}
\newcommand{\Rdist}{\ensuremath{R_{dist}}}
\newcommand{\Ssafe}{\ensuremath{S'_{safe}}}
\newcommand{\Scost}{\ensuremath{S'_{cost}}}
\newcommand{\grafoSR}{$G_{R,S'}$\xspace}
\newcommand{\PiS}{\ensuremath{\Pi_{S'}}}

\begin{document}
\title{Parameterized Complexity of $k$-Anonymity: Hardness and Tractability}
\author{Stefano Beretta\thanks{DISCo, Universit\`a degli Studi di
    Milano-Bicocca, Milano - Italy}
\and Paola Bonizzoni\thanks{DISCo, Universit\`a degli Studi di
    Milano-Bicocca, Milano - Italy}
\and Gianluca Della Vedova\thanks{Dipartimento di Statistica,
Universit\`a degli Studi di Milano-Bicocca, Milano - Italy}
\and Riccardo Dondi\thanks{Dipartimento di Scienze dei Linguaggi,
Universit\`a degli Studi di Bergamo, Bergamo - Italy}
\and Yuri Pirola\thanks{DISCo,
Universit\`a degli Studi di Milano-Bicocca, Milano - Italy}
}

\maketitle

\begin{abstract}
The problem of publishing personal data without giving up privacy is becoming
increasingly important. A clean formalization that has been recently proposed is the
$k$-anonymity, where
the rows of a table are partitioned in clusters of size at least $k$ and
all rows in a cluster become the same tuple,
after the suppression of some entries.
The natural optimization problem, where the goal is to minimize the number of
suppressed entries, is  hard even when the stored values are over a binary alphabet and 
as well as on a table consists of a bounded number of columns. 
In this paper we study how the complexity of the problem
is influenced by different parameters. 
First we show that the problem is W[1]-hard when  parameterized by the value
of the solution (and $k$). Then we exhibit a fixed-parameter algorithm when the problem is parameterized
by the number of columns and the maximum number of different values in any column.
Finally, we prove that
$k$-anonymity is still APX-hard even when restricting to instances with $3$ columns and $k=3$.
\end{abstract}

\section{Introduction}
In epidemic studies the analysis of
large amounts of personal data is essential. At the same time the
dissemination of the results of those studies, even in a compact and
summarized form, can provide some information that can be exploited to
identify the row pertaining to a certain individual. For instance, ZIP code, gender and date of birth can  uniquely
identify 87\% of individuals in the U.S.~\cite{Sweeny02}.
Therefore when managing personal data it is of  the utmost
importance to effectively protect individuals' privacy.

One approach to deal with such problem is
the $k$-anonymity model~\cite{DBLP:conf/pods/SamaratiS98,Sweeny02,DBLP:journals/tkde/Samarati01,Meyerson_PODS04}.
Each row of a given table represents all data regarding a certain individual.
Then different rows are clustered together, and
some entries of the rows in each cluster are suppressed (i.e.  they are
replaced with a $*$) so that
each cluster consists of at least $k$ identical rows.
Therefore each
row $r$ in the resulting table is clustered with at least other  $k-1$ rows
identical to $r$, hence the resulting data do not allow to identify any individual.
While such formulation is not really sophisticated and has some practical
limitations, it is definitely interesting from a
theoretical point of view, as witnessed by the rich literature available.
We will focus on
separating  the cases that can be
solved efficiently from those that are intractable, therefore hinting at which
strategies are  likely or not going to be successfully employed when studying
more sophisticated formalizations.
Notice that different formulations of the problem have also been
proposed~\cite{Aggar_PODS06}, for example allowing the generalization of
entry values,
that is an entry value can be replaced with a less specific
value~\cite{Aggar_JPT05}, or considering  a notion of
proximity among values~\cite{DBLP:conf/wads/DuEGL09}.


A parsimonious principle leads to the optimization problem where we want to
minimize  the
number of entries in the table to be suppressed.
The $k$-anonymity problem is known to be APX-hard
even  when the matrix entries are over a \emph{binary} alphabet and
$k=3$~\cite{conf/FCT/BonizzoniDD09}, as well as when the matrix has $8$
columns and $k=4$ (this time on arbitrary alphabets)~\cite{conf/FCT/BonizzoniDD09}.
Furthermore, a polynomial-time $O(k)$-approximation algorithm
on arbitrary input alphabet, as well as approximation algorithms
for restricted cases are known~\cite{Aggar_ICDT05}.
Recently, two polynomial-time approximation algorithms with factor $O(\log k)$ have been
independently proposed~\cite{Park_SIGMOD07,Gionis_ESA07}.

In this paper we investigate the parameterized complexity~\cite{ParameterizedComplexity,Niedermeier}
of the problem,
unveiling how different parameters are involved in the complexity of the problem.
A first systematic study of the parameterized complexity of the $k$-anonymity problem
has been proposed in~\cite{DBLP:conf/cocoa/ChaytorEW08}. 
Here, we follow the same direction,
showing that the problem is W[1]-hard when parameterized by the size
of the solution and $k$, and we provide a fixed-parameter algorithm, when the
problem is parameterized by the number of
columns and the maximum number of different values in any column.
These problems were left open in~\cite{DBLP:conf/cocoa/ChaytorEW08}.

In Table~\ref{tab:par-compl-status} we report the status of the parameterized complexity
of the $k$-anonymity problem, where in bold we have emphasized the new results
presented in this paper. We recall
that a problem $P$ parameterized by a set $Y$ of parameters
is in the class FPT~\cite{ParameterizedComplexity} if it admits an exact
algorithm with complexity $f(Y)n^{O(1)}$,
where $f$ is an arbitrary function, and $n$ is the size of the
input problem, while it is W[i]-hard~\cite{ParameterizedComplexity}, for some $1 \leq i \leq p$
if it is unlikely to be fixed-parameter tractable.
We recall that XP~\cite{ParameterizedComplexity} is a
superclass of all sets W[$p$]. Moreover, proving that a  problem $\Pi$ with parameter set $S$
is NP-hard when all parameters in $S$ are some constants, implies that
$(\Pi,S) \notin $ XP unless P = NP.

\begin{table}[htb]
\begin{footnotesize}
\centering
\begin{tabular}{|c||c|c|c|c|}
\hline  & $-$ & $k$ & $e$ & $k,e$ \\
\hline  $-$ & NP-hard \cite{Meyerson_PODS04} & $\notin XP$~\cite{conf/FCT/BonizzoniDD09,Aggar_ICDT05} & \textbf{W[1]-hard} \emph{new} & \textbf{W[1]-hard} \emph{new}\\
\hline $|\Sigma|$ & $\notin XP$~\cite{conf/FCT/BonizzoniDD09} & $\notin XP$~\cite{conf/FCT/BonizzoniDD09} & ??? & ???  \\
\hline $m$ & $\notin XP$ for $m \geq 8$~\cite{conf/FCT/BonizzoniDD09}& $\notin XP$ for $m \geq 8$, $k\geq 4$~\cite{conf/FCT/BonizzoniDD09} & FPT~\cite{DBLP:conf/cocoa/ChaytorEW08} & FPT~\cite{DBLP:conf/cocoa/ChaytorEW08} \\
\hline $n$ & FPT~\cite{DBLP:conf/cocoa/ChaytorEW08} & FPT~\cite{DBLP:conf/cocoa/ChaytorEW08}  & FPT~\cite{DBLP:conf/cocoa/ChaytorEW08} &  FPT~\cite{DBLP:conf/cocoa/ChaytorEW08}\\
\hline  $|\Sigma|,m$ & \textbf{FPT} \emph{new} & FPT~\cite{DBLP:conf/cocoa/ChaytorEW08} & FPT~\cite{DBLP:conf/cocoa/ChaytorEW08} & FPT~\cite{DBLP:conf/cocoa/ChaytorEW08} \\
\hline $|\Sigma|,n$ & FPT~\cite{DBLP:conf/cocoa/ChaytorEW08} & FPT~\cite{DBLP:conf/cocoa/ChaytorEW08} & FPT~\cite{DBLP:conf/cocoa/ChaytorEW08}  &  FPT~\cite{DBLP:conf/cocoa/ChaytorEW08}\\
\hline
\end{tabular}
\caption{Summary of the parameterized complexity status of the $k$-anonymity problem;
$|\Sigma|$ represents the maximum number of different values in a column,
$m$ represents the number of columns, $n$ represents the number of rows,
$k$ represents the minimum size of a cluster, $e$ represents the size of the
solution.}
\label{tab:par-compl-status}
\end{footnotesize}
\end{table}

The rest of the paper is organized as follows.
In Section~\ref{sec:pre-sec} we introduce some preliminary definition and we give
the formal definition of the $k$-anonymity problem. In Section~\ref{sec:par-hard} we
show that the $k$-anonymity is W[1]-hard.
In Section~\ref{sec:fpt-algo} we give a fixed parameter algorithm, when the problem is parameterized by the size of the alphabet and the number of columns.
Finally, in Section~\ref{sec:AP-3-3} we show that the $3$-anonymity problem is APX-hard, even when the rows have length bounded by $3$.

\section{Preliminary Definitions}
\label{sec:pre-sec}

Let us
introduce some preliminary definitions that will be used in
the rest of the paper.
Given a graph $G=(V,E)$, and $V' \subseteq V$, the \emph{subgraph induced}
by $V'$ is denoted by $G[V']=(V', E')$, where $E' = E\cap (V'\times V')$.
A graph $G=(V,E)$ is cubic when each vertex in $V$ has degree three.


Given an alphabet $\Sigma$, a row $r$ is a vector of
elements taken from the set $\Sigma$, and the
$j$-th element of $r$ is denoted by $r[j]$.
Notice that it is equivalent to  consider  a row  as a vector over $\Sigma$ or as a string
over alphabet  $\Sigma$.
Let $r_1, r_2$ be two equal-length rows. Then $H(r_1,r_2)$ is the Hamming
distance of $r_1$ and $r_2$, i.e. $|\{i: r_1[i]\neq r_2[i]\}|$.
Let $R$ be a set of $l$ rows, then
a \emph{clustering} of $R$ is a partition $\Pi=(P_1,\dots , P_t )$ of $R$.
Given a clustering $\Pi = (P_1,\dots , P_t )$ of $R$, we define the \emph{cost}
of the row $r$ belonging to a set $P_i$ of $\Pi$ as
$c_{\Pi}(r) = |\{j:\exists r_1,r_2\in P_i,\ r_1[j]\neq r_2[j]\}|$, that is the number of entries
of $r$ that have to be suppressed so that all rows in $P_i$ are identical.
Similarly we define the cost of a set $P_i$,  denoted by $c_{\Pi}(P_i)$, as
$|P_i||\{j:\exists r_1,r_2\in P_i,\ r_1[j]\neq r_2[j]\}|$.
The cost of $\Pi$, denoted by $c(\Pi)$, is defined as
$\sum_{P_i\in \Pi} c(P_i)$.
Given a set $S \subseteq R$ and a clustering $\Pi$ of $R$,  the
cost induced by $\Pi$ in set $S$ is
$c_{\Pi}(S) = \sum_{r \in S} c_{\Pi}(r)$.
Notice that, given a clustering $\Pi = (P_1,\dots, P_t)$ of $R$, the quantity
$|P_i|\max_{r_1,r_2\in P_i}\{H(r_1,r_2)\}$ is a lower bound for $c(P_i)$,
since all the positions for which $r_1$ and $r_2$ differ will
be deleted in each row of $P_i$.
We are now able to formally define the $k$-Anonymity Problem
($k$-AP). 

\begin{Problem}$k$-AP.\\
\textbf{Input}: a set $R$ of equal lenght rows over an alphabet $\Sigma_R$.\\
\textbf{Output}: a clustering $\Pi = (P_1,\dots, P_t)$ of $R$
such that for each set $P_i$, $|P_i|\geq k$ and $c(\Pi)$ is minimum.
\end{Problem}

In what follows, given a set $S$ of parameters, we denote by $\langle S \rangle$-AP
the $k$-AP problem parameterized by $S$, thus omitting $k$.
We will consider the following parameters:
$m$ is the number of columns of the rows in $R$;
$n$ is the  number of rows  in $R$;
$|\Sigma|$ is the  maximum number of different values in any column of the table;
$k$ is  the minimum size of a cluster;
$e$ is  the maximum number of entries that can be suppressed.

Let $\Pi=(P_1, \dots, P_z)$ be a solution of the $k$-AP problem.
Notice that a suppression at position $j$ of a row $r$ is represented
replacing the symbol $r[j]$ with a $*$.
Given a set
$P_j$ of $\Pi$, some entries of the rows clustered in $P_j$ are suppressed, so that
the resulting rows are all identical to a vector $r$ over alphabet $\Sigma_R \cup \{ * \}$; such
a vector is the \emph{resolution vector} associated
with $P_j$.
Given a
resolution vector $r$, we define $del(r)$ as the number of entries
suppressed in $r$, that is $del(r)=|\{j: r[j]=* \}|$.
Given a resolution vector $r$ and a row $r_i \in R$,
we say that $r$ is \emph{compatible} with row $r_i$ iff
$r[j] \neq r_i[j]$ implies $r[j]=*$.
Given a row $r_i$ of $R$ and a set of resolution vectors $S'$,
we define the set $comp(r_i,S')=\{ r \in S': r  \text{ is compatible with } r_i \}$.

%

Given a set $R$ of rows, we define a \emph{group} of rows of $R$ as a maximal
set of identical rows. Given a group $g$, the \emph{representative row} of $g$, denoted by $r(g)$, is
any row of $g$, while  $s(g)$ is the number of rows in $g$ and  $exc(g)=\max\{0, s(g)-k\}$.
A set $R$ of rows can be partitioned in groups of identical rows
in polynomial time~\cite{DBLP:conf/cocoa/ChaytorEW08}, therefore
we can compute in polynomial time
whether a set $R$ of rows is $k$-anonymous, i.e. $R$ can be partioned into groups of size at least $k$. 
If this is not possible, then observe that at least $k$ entries of $R$ must be suppressed 
to get a solution of the $k$-AP problem, that is $e \ge k$.
%
Hence $\langle e \rangle$-AP is in FPT iff $\langle e,k \rangle$-AP is in FPT.
%
Consequently our
parameterized reduction~\cite{ParameterizedComplexity,Niedermeier} will show
the fixed-parameter intractability of $\langle e \rangle$-AP and
$\langle e,k \rangle$-AP.

\section{$\langle e \rangle$-AP and $\langle e,k \rangle$-AP are W[$1$]-hard}
\label{sec:par-hard}

We show that $\langle e \rangle$-AP and $\langle e,k \rangle$-AP are
W[1]-hard. 
Given an set $R$ of equal length rows,
$\langle e \rangle$-AP and $\langle e,k \rangle$-AP ask if there exists
a clustering $\Pi = (P_1,\dots , P_t)$ of R such that $|P_i| \geq k$ for each set $P_i$,
and $c(\Pi) \leq e$.
We present a parameter preserving reduction from the $h$-Clique problem, which is
known to be W[1]-hard~\cite{DBLP:journals/tcs/DowneyF95}, to the $\langle e \rangle$-AP problem.
Given a graph $G=(V,E)$, an $h$-clique is a set $V' \subseteq V$
where each pair of vertices in $V'$ are connected by an edge of $G$, and
$|V'|=h$.
The $h$-Clique problem asks for a subset $V'$ of the
vertices of a given graph $G$ inducing an $h$-clique in $G$.

Clearly the vertices of a $h$-clique are connected by ${h} \choose{2}$ edges.
Given a graph $G=(V,E)$,  we use $m_G$  and $n_G$  to denote
respectively the number of edges and of vertices of $G$.
We construct the instance  $R$ of $\langle e \rangle$-AP associated with $G$.
First, let us define $k=2h^2$. The set  $R$ consists of
$(k+1)m_G+(k-{h \choose 2})$ rows
and $2h+n_G$ columns over alphabet $\Sigma_R= \{  0,1 \} \cup \{  \sigma_{i,j}:(v_i,v_j) \in E \}$.
More precisely, for each edge $e(i,j)=(v_i,v_j)$ in $E$, there is a
group $R(i,j)$ of $k+1$ identical rows $r_x(i,j)$, $1 \leq x \leq k+1$, where

\begin{itemize}
\item $r_x(i,j)[l]= \sigma_{i,j}$, for $1 \leq l \leq 2h$;
\item $r_x(i,j)[2h+i]= 1$, $r_x(i,j)[2h+j]= 1$;
\item $r_x(i,j)[2h+l]= 0$, for $l \neq i,j$ and $1 \leq l \leq n$.
\end{itemize}

Moreover, $R$ also contains a group $R_0$
made of $k-{h \choose 2}$ identical rows equal to $0^{2h+n_G}$.
%

\begin{Lemma}
\label{lem:W[1]-hard:prop-base}
Let $R$ be the instance of $\langle e \rangle$-AP associated with  $G$ and
consider two rows $r,r_x(i,j)$ of $R$, such that
$r \in R_0$ and $r_x(i,j) \in R(i,j)$. Then, $r[t] \neq r_x(i,j)[t]$,
for each $1 \leq t \leq 2h$.
\end{Lemma}


\begin{Lemma}
\label{lem:W[1]-hard-1}
Let $G=(V,E)$ be a graph, let $V'$ be a  $h$-clique of $G$ and let
$R$ be the instance of $\langle e \rangle$-AP associated with  $G$.
Then we can compute in polynomial time a solution $\Pi$ of $\langle e \rangle$-AP over instance $R$
with cost at most $6h^3$.
\end{Lemma}

\begin{Lemma}
\label{lem:W[1]-hard-2}
Let $G=(V,E)$ be an  instance of $h$-Clique, let
$R$ be the instance of $\langle e \rangle$-AP associated with $G$ and let  $\Pi$
be a solution of $\langle e \rangle$-AP over instance $R$ with cost at most $6h^3$.
Then we can compute in polynomial time a  $h$-clique  $V'$ of $G$.
\end{Lemma}
\begin{proof}
First we will prove that $\Pi$ must have a set $R_0'\supset R_0$. Assume to
the contrary that in $\Pi$ there are two sets $A$, $B$ containing at least a
row of $R_0$. Notice that $|R_0|<k$ while $|A|,|B|\ge k$. 
Moreover, by Lemma~\ref{lem:W[1]-hard:prop-base}, all rows in $A$ or $B$  must
have suppressed the first
$2h$ entries,  which results in at least $4hk>6h^3$ suppressions,
contradicting the assumption on the cost of the solution.
Hence,  $R_0$ is properly contained in a set
$R_0'$ of $\Pi$, as $|R_0| < k$.
Moreover,
let $r'$ be a row of $R_0'\setminus R_0$ and let $r$ be a row of  $\in R_0$.
By Lemma~\ref{lem:W[1]-hard:prop-base} $r'[t] \neq r[t]$ for each column $t$,
$1 \leq t \leq 2h$,  therefore all entries in the
first $2h$ columns of each row in $R_0'$ must be suppressed.

Now, let us prove that, for each set $R(i,j)$ of $R$, there exists a set
$R'(i,j)$ of $\Pi$ such that $R'(i,j)\subseteq R(i,j)$. Assume to the contrary
that no such set $R'(i,j)$ exists, for a given $R(i,j)$. Then either $R(i,j)\subseteq
R_0'$ or there exists a row of $R(i,j)$  clustered together with a row of
$R(x,y)$ in $\Pi$, with $(x,y)\neq (i,j)$.
In the first case, that is  $R(i,j)\subseteq R_0'$, $|R_0'|\ge
2k+1-{h \choose 2}$, by construction
all entries of the first $2h$ columns of the rows in $R_0'$ must be suppressed,
resulting in at least $2h(4h^2-{h \choose 2})>6h^3$ suppressions
and thus contradicting the assumption on the cost of the solution. Consider now the second case,
that is there is a set $A$ in $\Pi$ containing at least a row of two different
sets $R(i,j)$ and $R(x,y)$ of $R$. 
Observe that given $r'  \in R'_0 \setminus R_0$  and $r \in R_0$, $r$ and $r'$ differ in the first $2h$ columns.
Thus the entries of the  first $2h$ columns of the rows of $R_0'$ must be
suppressed,
resulting in at least $4hk>6h^3$ suppressed entries
and thus contradicting the assumption on the cost of the solution. Hence, for
each set $R(i,j)$ of $R$, there exists a set
$R'(i,j)$ of $\Pi$ such that $R'(i,j)\subseteq R(i,j)$.

By our previous arguments we can assume that $\Pi$ consists of the clusters
$R_0'$ and $R'(i,j)$, for each $R(i,j) \in R$,
and that $|R(i,j)|-1 \leq |R'(i,j)| \leq |R(i,j)|$. Notice that only
$R_0'$ can contain some suppressed entries.
Also  $|R_0'|=k$, for otherwise we can improve the cost of $\Pi$ by
moving a row
in $R(i,j) \cap R'_0$ from $R'_0$ to  $R'(i,j)$.
Now let $E'$ be the set of edges
$(v_i,v_j)$ of $G$ such that a row of $R(i,j)$ is in $R_0'$ and let $V'$ be the
set of vertices incident on at least an edge in $E'$. Then we can  show that
$G[V']$ is a $h$-clique.
Notice that the entries
in the first $2h$ columns of $R_0'$ must be suppressed, as well as all columns
with index $2h+l$ such that $v_l\in V'$, since in those columns all rows in
$R_0$ have value $0$ while some row in $R'_0 \setminus R_0$ have value $1$. An immediate
consequence is that the overall number of suppressed entries is at least
$2hk+k|V'|$. Since, by hypothesis, the number of suppressed entries is at
most $6h^3=3kh$, then $|V'|\le h$. Notice that,
since $|R_0 | = k - {h \choose 2}$ and $|R'_0| = k$, then $R'_0 \setminus R_0$
contains exactly ${h \choose 2}$ distinct rows corresponding to edges in $E'$
incident on $V'$ vertices.
Hence $V'$ induces a $h$-clique in $G$.
\qed\end{proof}

From Lemma~\ref{lem:W[1]-hard-1} and~\ref{lem:W[1]-hard-2},  our  reduction is
parameter preserving, therefore
$\langle e \rangle$-AP and $\langle e,k \rangle$-AP are W[1]-hard.

\section{An FPT algorithm for \smAP}
\label{sec:fpt-algo}

In this section we present a fixed-parameter algorithm for the
\smAP problem, that is the instance of the AP problem, where
the number $m$ of columns and the maximum number $|\Sigma|$ of different values in any column
are two parameters.
Notice that $k$-AP parameterized by exactly one of $|\Sigma|$ or $m$
is not in FPT, as $k$-AP is APX-hard (hence NP-hard)
even when one of $|\Sigma|$ or $m$ is a constant~\cite{conf/FCT/BonizzoniDD09}.

Before giving the details of the algorithm, let us first introduce some preliminary
definitions.
Let $R$ be an instance of \smAP,
and for each column of $R$ with index $j$,  $1 \leq j \leq m$, let $\Sigma_j$
be the set of different values that the rows of $R$ have in column $j$.
Notice that $|\Sigma_j| \leq |\Sigma|$, for each $1 \leq j \leq m$.
Let $\Sigma_j^{*} = \Sigma_j \cup \{ * \}$
and $\Sigma^{*} = \Sigma  \cup \{ * \}$.
Assume $\Pi = \{P_1, \cdots, P_z\}$  is a feasible solution
of \smAP over instance $R$. The set $S'$ consisting of  a
resolution vector for each set $P_i \in \Pi$ is called
\emph{candidate set} for solution \smAP.
Let $S$ be the set of possible rows of length $m$ and having value over alphabet
$\Sigma^*_j$ for the position $j$, $1 \leq j \leq m$,
then $|S|$ is
bounded by $|\Sigma^*|^m$.  Given a candidate set $S'$,
notice that $S' \subseteq S$ and that
each row $r \in R$ must compatible with at least one resolution
vector in $S'$. 

Given a row $r$ and the set  $S'$ of resolution vectors, recall that we denote by $Comp(r,S')$
the set of resolution vectors of $S'$ compatible with $r$. Moreover,  given a
resolution vector $r' \in S'$,
we denote by $del(r')$  the number of suppressions in $r'$.
For each row $r\in R$ we define its weight as
$w(r)= \max_{r_x \in Comp(r,S')} \{m-del (r_x)\}$.
Notice that $w(r) = m$ whenever  $r$ is compatible with a row without
suppressions. Informally, the weight of a row is equal to
the maximum number of its entries that might be preserved in a solution where $S'$ is
the set of resolution vectors.
Finally, we define $W= \sum_{r \in R} w(r)$ and $w'(r_x)= W+ m-del(r_x)+1$
for each row $r_x\in S'$.
Notice that $w'(r_x) \geq \sum_{r \in R} w(r)$, for each $r_x \in R$.
The weights defined above will be used later in Section~\ref{subsect_build} to define the weight
function $w_h$.

Let us first describe the general idea of the algorithm.
Given a candidate set $S'$, the algorithm computes an optimal solution $\PiS$ associated with a candidate
set $S' \subseteq S$ (see Algorithm~\ref{alg:solving-AP}).
The algorithm consists of two main phases.
In the first phase (Section~\ref{subsect_build}), given the set $R$ of input rows and the candidate set $S'$, the algorithm builds a weighted bipartite graph $G_{S',R}$
associated with $R$ and $S'$.
In the second phase (Section~\ref{subsect_solution}) a solution
of \smAP is computed starting from a maximum weighted matching
of the graph $G_{S',R}$.
Section~\ref{subsect_opt_solution} is devoted to prove
that the solution computed by the algorithm is optimal.

\begin{algorithm}[t]
\label{alg:solving-AP}
\KwIn{An instance $R$ of \smAP made  of a set of $n$ rows, each one consisting
  of $m$ symbols, and an integer $e$}
\KwOut{a solution of \smAP over instance $R$,
if \smAP admits a solution that suppresses at most $e$ entries\;}
$S\gets$ the set of resolved vectors of length $m$, where each $j$-th symbol, $1 \leq j \leq m$, is taken from the  alphabet $\Sigma^*_j$\;
$W= \sum_{r \in R} w(r)$\;
\ForEach{subset $S'$ of $S$}{%
  \grafoSR $\gets $ the graph associated with $R,S'$\;
  $M\gets $ a maximum matching of \grafoSR;
  $w\gets $ the weight of $M$\;
   \uIf{$M$ is feasible and $w \geq (W+1) k |S'| + m |R^l_{dist} \cup R^l_{safe}|- e$}{\Return{the solution $\PiS(M)$ of $R$ associated with $M$\;}}
}
\Return{No such solution exists}
\caption{Solving \smAP}
\end{algorithm}

\subsection{Building the graph \grafoSR}
\label{subsect_build}
Let us consider a candidate set $S'$ of vectors
for an optimal solution of \smAP.
Since $S' \subseteq S$,  there exist
at most $2^{|\Sigma^{*}|^m}$ possible candidate sets of rows $S'$, therefore our
FPT algorithm computes each candidate set $S'$ and verifies if there exists a solution $\PiS$
with cost at most $e$.
In order to verify if such a solution exists, the algorithm builds a bipartite graph
\grafoSR, as described in this section.
The intuitive idea behind the graph is that
edges of the graph correspond to possible ways of assigning each
row in $R$ to a resolution vector $x \in S'$. Rows
assigned to the same resolution vector $x \in S'$ are clustered in the solution $\PiS$.

The construction of the vertex set of the graph is based on  a a
partition of $R$ into two disjoint sets called $\Rsafe$ and
$\Rdist$ (that is $\Rdist = R \setminus \Rsafe$). The set $\Rsafe
$ consists of those rows $r \in R$ belonging to the group $g$ such that:
%
$s(g) \geq k$, that is $r$ belongs to a group of at least $k$ identical rows, and
there exists a row $r_j \in S'$, such that $r_j$ and $r(g)$ are the
same vector. 
%
Notice that only rows in \Rsafe{} might have no
suppressed entry in a solution $\PiS$.

%




The vertex set of \grafoSR$=(V,E)$ has $6$ sets.
Two sets ($\Rdist^l$, $\Rdist^r$) consist of vertices associated
with the rows in $\Rdist$, three sets ($\Rsafe'^l$, $\Rsafe^l$,
$\Rsafe^r$) consist of vertices associated
with the rows in $\Rsafe$, and a final set called $T$ consists of vertices
associated with the rows in $S'$. In the latter case notice that for each row
$x$ in $S'$ there exist $k$ vertices in $T$ to ensure that
the cluster associated with $x$ has size at least  $k$.
The vertex set is defined as follows:
%

%

\begin{itemize}
\item for each row $x\in \Rdist$, there is a corresponding vertex
$\Rdist^l(x)$ in $\Rdist^l$  and a corresponding vertex $\Rdist^r(x)$ in $\Rdist^r$;
\item for each  group $g$ consisting of the set of rows $\{ x_1, x_2 , \dots, x_{s(g)}  \}$,
where each $x_i \in \Rsafe$, $1 \leq i \leq s(g)$,
there are $k$ corresponding vertices in $\Rsafe'^l$, (such vertices
are denoted by $\Rsafe'^l(g,1), \ldots , \Rsafe'^l(g,k)$),
$exc(g)$ corresponding vertices in $\Rsafe^l$ (such vertices
are denoted by $\Rsafe^l(g,1),$ $\ldots ,$ $\Rsafe^l(g,exc(g)$), and
$exc(g)$ corresponding vertices in $\Rsafe^r$ (such vertices are denoted by
$\Rsafe^r(g,1),$ $\ldots ,$
$\Rsafe^r(g,exc(g)$); 
 \item for each row $x \in S'$, there are $k$ corresponding vertices in $T$ (such
vertices are denoted by $T(x, 1), \dots, T(x, k)$).
\end{itemize}


Notice that  our graph \grafoSR is edge-weighted. Let $w_h$ be the
weight function assigning a positive weight to each edge of
\grafoSR. Given the  set of edges $E' \subseteq E$, we denote by
$w_h(E')= \sum_{e \in E'} w_h(e)$.
%
%

First, notice that the set $S'$  consists of two disjoint sets:
the set \Ssafe consists of those rows in $S'$ that have no
suppressions, while $\Scost=S' \setminus \Ssafe$.
Each edge connects a vertex of $\Rsafe'^l \cup \Rsafe^l \cup
\Rdist^l$ with a vertex of $\Rsafe^r \cup \Rdist^r \cup T$, hence
the graph \grafoSR is bipartite.
The set $S'$  consists of two disjoint sets:
the set \Ssafe consists of those rows in $S'$ that have no
suppressions, while $\Scost=S' \setminus \Ssafe$.
Intuitevely, each edge represents a possible
assignment of a row in $R$ to a resolution vector in $S'$.

\begin{algorithm}
\label{alg:matching} \KwIn{A graph \grafoSR associated with an
instance $R$ and a maximum weight matching $M$ of \grafoSR}
\KwOut{A solution $\PiS(M)$ of \smAP over instance $R$}
\ForEach{edge $y$ of $M$}{%
  \uIf(\tcc*[f]{edges defined at point 1}){$y = (\Rdist^l(r),T(x,j))$}{%
    row $r$ is assigned to
    a set whose resolution row is $x$, $x \in S'$}
  \uIf(\tcc*[f]{edges defined at point 2}){$y = (\Rdist^l(r),\Rdist^r(r))$}{%
    row $r$ is assigned to
    a set whose resolution row is $r_y= \text{arg } \max w(r)$, $r_y \in S'$\;}
  \uIf(\tcc*[f]{edges defined at point 3}){$y = (\Rsafe'^l(g,i), T(x,j))$}{%
    assign the $i$-th row of $g$ to a set whose resolution row is
    $x$, $x \in S'$\;
  }
  \uIf(\tcc*[f]{edges defined at point 4}){$y = (\Rsafe^l(g,i), T(x,j))$}{%
    assign the $i$-th exceeding row of $g$ to a set whose resolution row is
    $x$, $x \in S'$\;
  }
  \uIf(\tcc*[f]{edges defined at point 5}){$y = (\Rsafe^l(g,i), \Rsafe^r(g,i))$}{%
    assign the $i$-th exceeding row of group $g$ to the set whose
    resolution row is $r(g)$, with $r(g) \in S'$ and $r \in \Rsafe$\;
  }
} \caption{From a matching to a feasible solution of \smAP.}
\end{algorithm}

Now we are ready to define formally the set of edges $E$ of
\grafoSR and the weight function $w_h$. There are five possible
kinds of edges.

\begin{enumerate}
\item Let $r$ be a row of $\Rdist$, and let $x$ be a row in  $Comp(r,S') \cap \Scost$.
Then there is an edge $y=(\Rdist^l(r),T(x,j))$, for each $1 \leq j \leq k$, with weight
$w_h \left( y \right)= w'(x)$.

\item Let $r$ be a row in $\Rdist$.
Then there is an edge $y=(\Rdist^l(r),\Rdist^r(r))$ with weight
$w_h \left( y \right)=w(r)$.

\item
Let $g$ be a group consisting of rows $\{ r_1, \dots , r_{s(g)} \}$, where
$r_i$, for each $i$ with $1 \leq i \leq s(g)$, is a row of $\Rsafe$;
let $r'$ be the resolution vector of $\Ssafe$ identical to $r(g)$.
Then there is an edge $y_i=(\Rsafe'^l(g,i), T(r',i))$, for each $i$ with $1 \leq
i \leq k$. All edges $y_i$ have weight $w_h \left( y_i \right)= w'(r')$.

\item Let $g$ be a group consisting of rows $\{ r_1, \dots , r_{s(g)} \}$, where
$r_i$, for each $i$ with $1 \leq i \leq s(g)$, is a row of $\Rsafe$; let $x$ be a row in $Comp(r(g),S') \cap \Scost$.
Then there is an edge $y_{i,j}=(\Rsafe^l(g,i), T(x,j))$, for each $i$ with $1 \leq
i \leq exc(g)$ and for each $j$ with $1 \leq j \leq k$. All edges $y_{i,j}$ have weight $w_h \left( y_{i,j} \right)= w'(x)$.

\item Let $g$ be a group consisting of rows $\{ r_1, \dots , r_{s(g)} \}$, where
$r_i$, $1 \leq i \leq s(g)$, is a row of $\Rsafe$.
Then there is an edge $y_i=(\Rsafe^l(g,i), \Rsafe^r(g,i))$ for each $i$
with $1 \leq i \leq exc(g)$. All edges $y_i$ have weight $w_h \left( y_i \right)= w(r(g))$.
\end{enumerate}

%

\subsection{Computing a solution of \smAP}
\label{subsect_solution}
In this  section we  prove  in
Lemma~\ref{lem:feas:solution} that  $\PiS(M)$ is a clustering of
the rows in $R$ that is a feasible solution for the \smAP problem. See Fig. \ref{fig:ex-alg} for an example.

Since \grafoSR bipartite, we can efficiently compute a maximum weight matching $M$ of
\grafoSR~\cite{fastmwb}.
Given a matching $M$ of the graph \grafoSR,  Algorithm~\ref{alg:matching}
computes in polynomial time a clustering $\PiS(M)$ of the rows in $R$.
Informally, the clustering is computed by assigning
the rows in $R$ to the resolution vector in $S'$,
using the edges in the matching $M$.

%
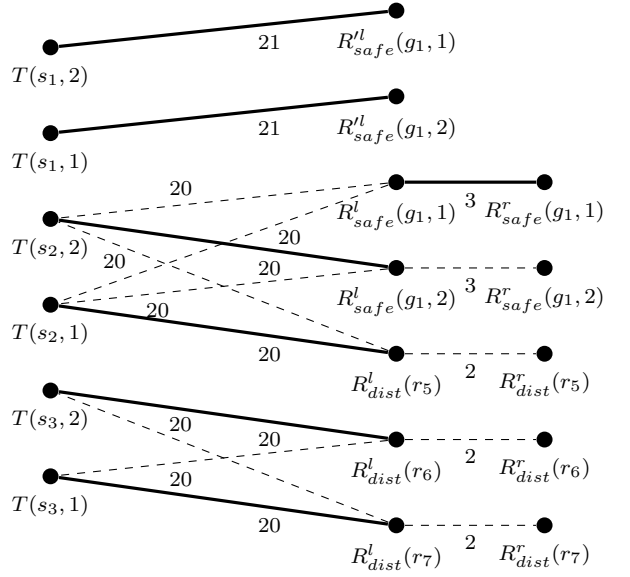
\begin{figure}[tb]
\label{fig:ex-alg}
\begin{minipage}{4cm}\small\begin{center}
\begin{tabular}{clll}
\multicolumn{4}{c}{Rows $R$}\\\hline\hline
Name&Data&$w$&Group\\\hline
$r_1$ & aaa&$3$&\\\cline{1-3}
$r_2$ & aaa&$3$&\\\cline{1-3}
$r_3$ & aaa&$3$&\raisebox{1.5ex}[0pt]{$g_1$}\\\cline{1-3}
$r_4$ & aaa&$3$&\\\hline
$r_5$ & aba&$2$&$g_2$\\\hline
$r_6$ & bbb&$2$&$g_3$\\\hline
$r_7$ & bbc&$2$&$g_4$\\\hline
\end{tabular}\vspace{2em}

\begin{tabular}{cll}
\multicolumn{3}{c}{Resolution vectors $S'$}\\\hline\hline
Name&Vectors&$w$\\\hline
$s_1$ & aaa&$21$\\\hline
$s_2$ & a*a&$20$\\\hline
$s_3$ & bb*&$20$\\\hline
\end{tabular}

\end{center}
\end{minipage}\hfill
\begin{minipage}{8cm}
\tikzstyle{vertex}=[circle,fill=black,minimum size=6pt,inner sep=0pt]
\begin{tikzpicture}[scale=0.65,font=\scriptsize]
\foreach \name/\x in {1/7, 2/6, 3/5}
     \node[vertex] (L-\name) at (7,1.75*\name) [label=below:$R^l_{dist}(r_{\x})$]{};

\foreach \name/\x in {4/2, 5/1}
     \node[vertex] (L-\name) at (7,1.75*\name) [label=below:$R^l_{safe}({g_1,\x})$]{};

\foreach \name/\x in {6/2, 7/1}
     \node[vertex] (L-\name) at (7,1.75*\name) [label=below:$R'^l_{safe}({g_1,\x})$]{};

\foreach \name/\x in {4/2, 5/1}
     \node[vertex] (R-\name) at (10,1.75*\name) [label=below:$R^r_{safe}({g_1,\x})$] {};

\foreach \name/\x in {1/7, 2/6, 3/5}
     \node[vertex] (R-\name) at (10,1.75*\name) [label=below:$R^r_{dist}(r_{\x})$] {};

\foreach \name/\x/\y/\z in {1/1/s_3/1, 2/2/s_3/2, 3/3/s_2/1, 4/4/s_2/2, 5/5/s_1/1, 6/6/s_1/2}
     \node[vertex] (S-\name) at (0,1+1.75*\x) [label=below:{$T(\y,\z)$}] {};


\foreach \from/\to/\w in {1/1/2, 2/2/2, 3/3/2, 4/4/3}
\draw [-, thin, dashed] (L-\from) --node[midway,below] {\w} (R-\to);

\foreach \from/\to in {2/1}
\draw [-, thin, dashed] (L-\from) -- node[pos=0.63,below] {20} (S-\to);

\foreach \from/\to in {3/4}
\draw [-, thin, dashed] (L-\from) -- node[pos=0.83,below] {20} (S-\to);

\foreach \from/\to in {4/3}
\draw [-, thin, dashed] (L-\from) -- node[pos=0.7,below] {20} (S-\to);

\foreach \from/\to in {5/3}
\draw [-, thin, dashed] (L-\from) -- node[pos=0.3,below] {20} (S-\to);

\foreach \from/\to in {1/2, 5/4}
\draw [-, thin, dashed] (L-\from) -- node[pos=0.63,above] {20} (S-\to);

\foreach \from/\to in {5/5}
\draw[very thick] (L-\from) -- node[midway,below] {3} (R-\to);

\foreach \from/\to/\w in {1/1/20, 2/2/20, 3/3/20, 4/4/20, 6/5/21, 7/6/21}
\draw [very thick] (L-\from) -- node[pos=0.36,below] {\w} (S-\to);
\end{tikzpicture}
\end{minipage}
\caption{An instance $R$ of \smAP, with $k=2$ and $m=3$, a resolution vector
  set $S'$ and the associated
  graph \grafoSR. The thick edges are a maximum weight matching of \grafoSR.
The corresponding solution is made of the sets $\{r_1, r_2, r_3\}$ (cost $0$), $\{r_4, r_5\}$ (cost $2$),
$\{ r_6 ,r_7 \}$ (cost $2$).}
\end{figure}

Notice that,
each vertex $\Rsafe^l(r,i)$ has only the edge
$(\Rsafe^l(r,i),T(r,i))$ on it, hence we can always add those edges to any
matching%
\footnote{Notice that these
  connected components are introduced only to simplify the
relationship between a matching $M$ and the corresponding solution
$\PiS(M)$ of \smAP}.
%
Let $M$ be a matching of \grafoSR and let $v$ be a vertex of \grafoSR, then
we say that $v$ is \emph{covered} by a matching $M$
if there exists an edge of $M$ for which $v$ is one
of its endpoints.
Moreover, we will say that $M$ is  \emph{feasible} if all vertices in $T$ are covered by $M$.
When a  matching $M$ covers all vertices in $R^l_{dist} \cup
R^l_{safe}$ and is feasible, it is defined as a \emph{complete
matching}.
Let $\Pi$ be a clustering of an instance $R$ of the \smAP problem. Then
$\Pi$ is \emph{feasible} if and only if each set of the partition
$\Pi$ contains at least $k$ rows.
The next part of this section is devoted to show that every maximum weight matching $M$
is complete and that clustering $\PiS(M)$ is \emph{feasible}. 
%
First, we will show in the next two lemmata that, given $W' =
k\sum_{r_x\in T} w'(r_x)$, $W'$ is a threshold that distinguishes
between matchings that are feasible and those that are not.

\begin{Lemma}
\label{lem:eq-matching-weight-T} Let $M$ be a matching  of
\grafoSR, let $X$ be the subset of $T$ consisting of the vertices
of $T$ that are covered by $M$, and let $M_1$ be the subset of the
edges of $M$ that have one endpoint in $X$. Then the total weight
of the edges in $M_1$ is exactly $\sum_{T(t,i)\in X} w'(t)$.
\end{Lemma}
\begin{proof}
It is an immediate consequence of the observation that all edges where an
endpoint is $T(t,j)$
have the same weight $w'(t)$, with $t \in S'$.
\qed\end{proof}

\begin{Lemma}
\label{lem:eq-matching} Let $M$ be a matching  of  \grafoSR and
let $M_1$ be the subset of the edges of $M$ that have one endpoint
in $T$. Then the total weight of the edges in $M_1$ is at least
$W' = k\sum_{r\in S'} w'(r)$  if and only if $M$ is feasible.
\end{Lemma}
\begin{proof}
Let $M_1$ be the subset of the edges of $M$ that have one endpoint
in $T$, and let $W_1$ be the total weight of edges in $M_1$. An
immediate consequence of Lemma~\ref{lem:eq-matching-weight-T} is
that $W_1 = W'$ if and only if  $M_1$ is feasible.
Assume now that $M$ is not feasible, then there exists at least one vertex
$S'(x,j)\in T$ that is not covered by $M$. Again, a consequence of
Lemma~\ref{lem:eq-matching-weight-T} is that $W_1 \le W'-w'(x)$.
Let $M_2$ be the set $M\setminus M_1$. By construction, $w'(x)> W$ and $W$
is an upper bound on the total weight of $M_2$, therefore $W_1 + w_h(M_2) < W'$,
completing the proof.
\qed\end{proof}

Using Lemmata~\ref{lem:eq-matching-weight-T} and \ref{lem:eq-matching},
we can prove Lemma~\ref{lem:feas:solution}.

\begin{Lemma}
\label{lem:feas:solution}
Let $M$ be a maximum weight matching of \grafoSR, then $M$ is complete and
the solution $\PiS(M)$ computed by Algorithm~\ref{alg:matching} is
feasible.
\end{Lemma}


\subsection{Proving the optimality of $\PiS(M)$}
\label{subsect_opt_solution}

This section is devoted to prove that, starting from a maximum
weight matching $M$, Algorithm~\ref{alg:matching} computes an optimal
solution $\PiS(M)$ of \smAP.
In order to prove that any maximum weight matching $M$ of
the graph \grafoSR 
leads to an optimal solution of \smAP over instance $R$,  we are going
to prove that
$\sum_{(u,v) \in M} w_h\left((u,v)\right) \geq (W+1)k|S'|+m|R^l_{dist} \cup R^l_{safe}
\cup R'^l_{safe}| - e$
if and only if \smAP over instance $R$ admits a solution with
cost not greater than $e$, and such solution is computed by applying
Algorithm~\ref{alg:matching}.
Such result will be obtained through a sequence of technical lemmata.

Since $M$ is a maximum weighted matching,
we can assume by Lemma~\ref{lem:feas:solution} that $M$ is complete.
Given a complete matching $M$, we denote by $M(T)$
the set of  edges of $M$ with one endpoint in
$R^l_{dist} \cup R^l_{safe} \cup R'^l_{safe}$ and one endpoint in $T$,
while we denote by $M(L)$ the set of those edges of $M$ that have
one endpoint in
$R^l_{dist} \cup R^l_{safe}$ and one endpoint in $R^r_{dist} \cup R^r_{safe}$.
Furthermore, let us denote by $V(T)$ the set of vertices of
$R^l_{dist} \cup R^l_{safe} \cup R'^l_{safe}$ that are endpoints of an edge in
$M(T)$ and by $V(L)$ the set of vertices of
$R^l_{dist} \cup R^l_{safe}$ that are endpoints of an edge in
$M(L)$. Notice that by definition of $V(L)$ and, by definition of
complete matching,
$V(T) \cup V(L) = R^l_{dist} \cup R^l_{safe} \cup R'^l_{safe}$.
Finally, let us denote by $R(L)$ the set of rows in $R$
associated with the vertices in $V(L)$.
Lemma~\ref{lem:cost:match:part:feas} shows how the
weight of a complete matching $M$ is related to the edge weights of
\grafoSR.

\begin{Lemma}
\label{lem:cost:match:part:feas}
Let $M$ be a complete matching of \grafoSR,
and let $w_h(M)$ be the total weight of $M$. Then
$w_h(M) = k \sum_{r \in S'} ( W + m -del(r)+1) + \sum_{r \in R(L)} (m - del(r)) =
(W+1) k |S'| + m |R^l_{dist} \cup R^l_{safe} \cup R'^l_{safe}|- (k\sum_{r \in S'} del(r)
+ \sum_{r \in R(L)} del(r))$.
\end{Lemma}

In the next two lemmata, we will show that: (i) given an instance $R$
of \smAP, if there exists a solution of \smAP over $R$
that suppresses at most $e$ entries then
the graph \grafoSR associated with $R$ admits
a complete matching of \grafoSR with total weight
$w_G(M) \geq (W+1) k |S'| + m |R^l_{dist} \cup R^l_{safe} \cup R'^l_{safe}|- e$;
(ii) given a complete matching of the graph \grafoSR of total weight
$w_G(M) \geq (W+1) k |S'| + m |R^l_{dist} \cup R^l_{safe} \cup R'^l_{safe}|- e$,
Algorithm~\ref{alg:matching} returns a solution $\PiS(M)$
of \smAP that suppresses at most $e$ entries.
These lemmata, coupled with Lemma~\ref{lem:feas:solution},
prove the correctness of Algorithm~\ref{alg:matching} in
Theorem~\ref{theo:correctness-algorithm}.

\begin{Lemma}
\label{lem:from-partially-feasible-solution-to-matching}
Let $R$ be an instance of \smAP{}, let $\PiS$ be a feasible
solution of \smAP{} over instance $R$ that suppresses at most $e$ entries,
let \grafoSR be the graph associated with $R$ and $S'$.
Then there exists a complete  matching of \grafoSR with total weight
$w_G(M) \geq (W+1) k |S'| + m |R^l_{dist} \cup R^l_{safe} \cup R'^l_{safe}|- e$.
\end{Lemma}


\begin{Lemma}
\label{lem:feasible_matching-solution}
Let $R$ be an instance of \smAP, let \grafoSR be the graph associated with $R$, and
let $M$ be a complete matching of \grafoSR of weight
$w_h(M) \geq (W+1) k |S'| + m |R^l_{dist} \cup R^l_{safe} \cup R'^l_{safe}|- e$.
Then, starting from the
matching $M$ of \grafoSR, Algorithm~\ref{alg:matching} computes a
feasible solution $\PiS(M)$ of \smAP over instance $R$, where there are at
most $e$ suppressions.
\end{Lemma}
\begin{proof}
Since $M$ is complete, for each vertex $T(x,j)$ of $T$, with $1 \leq j \leq k$,
there exists an edge  $(v, T(x,j))\in M$ for some
$v \in (R^l_{dist} \cup R^l_{safe} \cup R'^l_{safe})$. Then
Algorithm~\ref{alg:matching} defines a solution $\PiS(M)$ for \smAP
assigning, for each edge $(v, T(x,j))$,
the row $r$ corresponding to vertex $v$
to the set that has resolution vector $x$.
More precisely, row $r$ is defined by
Algorithm~\ref{alg:matching} as the $j$-th element of the set
that has resolution vector $x$.
Therefore each set associated with a resolution row $x \in S'$ will consist of
at least $k$ rows compatible with $x$.
Hence  $\PiS(M)$ is a feasible solution.

Recall that $M$ has a total weight of at least
$(W+1) k |S'| + m |R^l_{dist} \cup R^l_{safe} \cup R'^l_{safe}|- e$.
We will prove that $\PiS(M)$ induces at most $e$ suppressions.
By  Lemma~\ref{lem:cost:match:part:feas},
$
w_h(M) = k \sum_{r \in S'} ( W + m -del(r)+1) + \sum_{r \in R(L)} m - del(r)
= (W+1) k |S'| + m |R^l_{dist} \cup R^l_{safe} \cup R^l_{safe}|- (k\sum_{r \in S'} del(r) + \sum_{r \in R(L)} del(r))
\geq (W+1) k |S'| + m |R^l_{dist} \cup R^l_{safe} \cup R^l_{safe}|- e
$
where 
$k \sum_{r \in S'} del(r) + \sum_{r \in R(L)} del(r) \leq e$.
Notice that, by definition of  $\PiS(M)$, each vertex of $V(T)$ corresponds to
a row in $R$ assigned to a set with a resolution vector in $S'$. Such rows
associated with $V(T)$ induce
a cost in $\PiS(M)$ of $k \sum_{r \in S'} del(r)$. Furthermore, the vertices of $V(L)$
corresponds to rows of $R$  inducing a cost of at most
$\sum_{r \in R(L)} del(r)$.
Therefore $\PiS(M)$ induces $k \sum_{r \in S'} del(r) + \sum_{r \in R(L)} del(r) \leq e$
suppressions.
\qed\end{proof}

\begin{Theorem}
\label{theo:correctness-algorithm}
Let $R$ be an instance of \smAP. Then
Algorithm~\ref{alg:solving-AP} returns a solution  $\PiS(M)$ of cost at most $e$ if and only if
such a solution exists.
\end{Theorem}
\begin{proof}
By Lemma~\ref{lem:feas:solution}, $\PiS(M)$ is feasible.
Hence if $\PiS(M)$ suppresses at most $e$ entries, then \smAP admits a solution
of cost at most $e$.
On the other hand, by
Lemma~\ref{lem:from-partially-feasible-solution-to-matching},
if there exists a solution $\Pi'$ of 
$R$ that suppresses
at most $e$ entries, then there exists a feasible matching $M$ with
weight $w_G(M) \geq (W+1) k |S'| + m |R^l_{dist} \cup R^l_{safe} \cup R'^l_{safe}|- e$.
Then, by Lemma~\ref{lem:feasible_matching-solution},
Algorithm~\ref{alg:solving-AP} returns a solution $\PiS(M)$ of \smAP
that suppresses at most $e$ entries.
\qed\end{proof}

If \smAP admits a solution that suppresses at most $e$ entries,
then there  exists a set $S^*$ of resolution vectors such that
$\Pi_{S^*}$ is a solution for \smAP with resolution vectors $S^*$
with the property that $\Pi_{S^*}$ suppresses at most $e$ entries.
Now, there exist $O(2^{(|\Sigma|+1)^m})$ possible sets of resolution vectors and
the construction of graph \grafoSR requires $O(k|S^*||R|) \leq O(k e |R|) \leq O(kmn^2)$.
A maximum matching $M$ of a bipartite graph
can be computed in polynomial time~\cite{fastmwb}
and starting from $M$, we can compute a solution of the \smAP
in time $O(|M|) \leq O(m)$. Hence the overall time complexity of the algorithm
is $O(2^{(|\Sigma|+1)^m}kmn^2)$.

\section{APX-hardness of $3$-AP($3$)}
\label{sec:AP-3-3}

In this section we investigate the computational and approximation complexity of
$3$-AP($3$), that is $k$-AP when each row
consists of exactly $3$ columns and $k=3$.
We show that $3$-AP($3$) is APX-hard
via an L-reduction
from Minimum Vertex Cover on Cubic Graphs (MVCC), which is known to be
APX-hard~\cite{AK92}.
Due to page limit, we only  sketch  the proof.
The MVCC problem, given a cubic graph $G=(V,E)$, asks for a smallest
$C \subseteq V$ such that each edge of $G$ has at least one
of its endpoints in $C$.

Let $G=( V , E )$ be instance of MVCC, where $|V|=n$ and $|E|=m$.
The reduction builds an
instance $R$ of $3$-AP($3$) associating with each vertex $v_i \in V$
a set $R_i$ consisting of $9$ rows, and with each edge $e=(v_i, v_j) \in E$ a set
$E_{i,j}$ consisting of $7$ rows.
Finally, a set $X$ of $3$ more rows is added to $R$.

Now we can describe formally our reduction.
Let $R_i$ be the set of rows associated with vertex $v_i \in V$.
The rows in $R_i$ have values over an alphabet
$\Sigma_i= \{ \sigma_i, \sigma_{i,1}, \sigma_{i,2}, \sigma_{i,3}  \}$.
The set $R_i$ consists of $9$ rows belonging to $6$ groups,
denoted by $g_1(v_i), \dots , g_6(v_i)$, of identical rows.
The representative rows of groups $g_1(v_i), \dots , g_6(v_i)$, and the cardinality of the groups,
are defined as follows:

\vspace*{0.12cm}
- $r(g_h(v_i))=\sigma_{i,h} \sigma_i \sigma_{i,h}$,
    with $h \in \{ 1,2,3 \}$; each group $g_h(v_i)$,  with $h \in \{ 1,2,3 \}$, consists of exactly two rows;

- $r(g_{3+h}(v_i))= \sigma_i \sigma_i \sigma_{i,h}$, with $h \in \{ 1,2,3 \}$;
each group $g_{3+h}(v_i)$, with $h \in \{ 1,2,3 \}$, consists of exactly one row.
\vspace*{0.12cm}

Notice that given two rows $r,r'$ belonging to different groups of $R_i$,
$H(r,r')=1$ iff $r \in g_h(v_i)$, $r' \in g_{3+h}(v_i)$ (or the converse) or
$r,r' \in \{ g_4(v_i), g_5(v_i), g_6(v_i)\}$.
Given a group $g_h(v_i)$, with $h \in \{ 1,2,3 \}$, each symbol $\sigma_{i,h}$
is called the \emph{private} symbol of  $g_h(v_i)$.
The groups of rows $g_j(v_i)$,  with $j \in \{ 1,2,3\}$, are denoted as the
\emph{docking groups} of $R_i$, and each of them is associated with a set $E_{i,h}$
of rows encoding an edge $(v_i,v_h)$ of $G$.
More precisely, given the set of rows $E_{i,j}$, we denote by
$d_{i,j}(g(v_i))$ the docking group of $R_i$ associated with set $E_{i,j}$.

Now, let us build the set $E_{i,j}$ of rows associated with an edge $(v_i,v_j)$.
Let $d_{i,j}(g(v_i))$ and $d_{i,j}(g(v_j))$ be the two docking groups of $R_i$ and $R_j$ respectively,
associated with the set $E_{i,j}$.
Let $\sigma_{i,x}$ and $\sigma_{j,y}$ be the private symbols of groups $d_{i,j}(g(v_i))$ and $d_{i,j}(g(v_j))$
respectively.
The set
$E_{i,j}$  consists of $7$ rows distributed in $6$ groups.
The rows of $E_{i,j}$ have values over alphabet
$\Sigma_{i,j} = \{  \sigma_{i,x}, \sigma_{j,y} ,  \sigma_{i,j} , \sigma_{i,j,4}, \sigma_{i,j,5},
\sigma_{i,j,6} \}$.
Let us define the representative rows and the cardinality of the groups in $E_{i,j}$:

\vspace*{0.15cm}
- $r(g_1(v_i,v_j))= \sigma_{i,x} \sigma_{i,j} \sigma_{i,x}$;
group $g_1(v_i,v_j)$ consists of a single row;

- $r(g_2(v_i,v_j))= \sigma_{i,x} \sigma_{i,j} \sigma_{j,y}$; group $g_2(v_i,v_j)$ consists of two rows;

- $r(g_3(v_i,v_j))= \sigma_{j,y} \sigma_{i,j} \sigma_{j,y}$;
group $g_3(v_i,v_j)$ consists of a single row;

- $r(g_t(v_i,v_j))= \sigma_{i,j,t} \sigma_{i,j} \sigma_{i,j,t}$, with $t \in \{ 4,5,6 \}$;
each group $g_t(v_i,v_j)$, with $t \in \{ 4,5,6 \}$, consists of a single row.

\vspace*{0.15cm}

The group of $E_{i,j}$ that has two occurrences of symbol $\sigma_{i,x}$ shared with
$d_{i,j}(g(v_i))$ is called the $i$-group of set $E_{i,j}$, and is denoted as
$g^i(v_i,v_j)$.
Notice that, given two rows $r,r'$ of $R_i$, $E_{i,j}$ respectively, then
$H(r,r')=1$ iff $r \in d_{i,j}(g(v_i))$ and $r' \in g^i(v_i,v_j)$.

Finally, a set $X$ of $3$ rows $x_1, x_2, x_3$ are added to $R$. The rows in
$X$ have values over an alphabet $\Sigma_x$ disjoint from any other set $\Sigma_i$, $\Sigma_{i,j}$.
Each row $x_i=w_i^3$, and it has Hamming distance $3$ from any other row of $R$.
Therefore for any set $C$ containing some rows $x_i$, all  positions
of a row in $C$ will be suppressed.

Now, consider the set $R_i$. The following lemma
gives a lower bound on the cost of an optimal solution of $3$-AP($3$) over instance
$R_i$.

\begin{Lemma}
\label{lem:opt-R_i}
Let $R_i$ be a set of rows, then an optimal solution of
$3$-AP($3$) over instance $R_i$ has a cost of at least $9$.
\end{Lemma}

The main idea of the reduction is showing that we can consider a set of solutions,
called \emph{canonical solutions}, that is solutions where:

\vspace*{0.15cm}
(i) $\Pi$ contains exactly one cluster $X$ containing only suppressed  entries;

(ii)  each set $R_i$ is associated with either a \tya or a \tyb solution
(to be defined later), eventually with the contribution of
some rows in the sets $E_{i,j}$ for a \tyb solution;

(iii) two sets $R_i$, $R_j$ are associated with a \tyb solution
only if there is no edge set $E_{i,j}$ in the instance $R$, that is
the corresponding vertices $v_i$, $v_j$ are not adjacent in $G$;

(iv) either an edge set is part of a \tyb solution of some set $R_i$ and has a total cost of $10$ or
it has a total cost of $11$.

\vspace*{0.15cm}
%
Notice that, by construction, in a canonical solution, rows $x_1, x_2, x_3 \in X$.
%
%
%

Let us define the notions of \tya and \tyb solution. Given a set $R_i$ and the
edge sets $E_{i,j}$, $E_{i,h}$, $E_{i,l}$, a \tya solution for $R_i$
consists of three sets $S_{i,1}, S_{i,2}, S_{i,3}$,
where $S_{i,t}=g_t(v_i) \cup g_{t+3}(v_i)$, while
%
%
%
a \tyb solution
consists of the following sets:
(i) three sets $d_{i,j}(g(v_i)) \cup g^i(v_i,v_j)$,
$d_{i,h}(g(v_i)) \cup g^i(v_i,v_h)$, $d_{i,l}(g(v_i)) \cup g^i(v_i,v_l)$;
(ii) $g_4(v_i) \cup g_5(v_i) \cup g_6(v_i)$.
%

Lemma~\ref{lem:compute-canonical} is the main technical contribution of this section.
\begin{Lemma}
\label{lem:compute-canonical}
Let $\Pi$  be a solution of $3$-AP($3$) over instance $R$. Then we can compute in polynomial time a canonical solution $\Pi'$ of $3$-AP($3$) over instance $R$ such that $c(\Pi') \leq c(\Pi)$.
\end{Lemma}
\noindent
\emph{Sketch of the proof.}
By direct inspection, it is immediate to notice that \tya and \tyb
solutions induce $9$ suppression
in rows of $R_i$ hence, by Lemma~\ref{lem:opt-R_i}, they are optimal for $R_i$.
The next step is computing in polynomial time a solution $\Pi''$ such that
each set $R_i$  is associated
in $\Pi''$ only with either a \tya or \tyb solution, and such that $c(\Pi'')
\leq c(\Pi)$.
Such step is obtained by exploiting the optimality of \tya and \tyb solutions
for $R_i$, and some properties of the instance $R$.

Then, starting from such solution $\Pi''$, we can compute in polynomial time
a canonical solution $\Pi'$ such that
$c(\Pi') \leq c(\Pi'')$. The main idea to prove this result is that for any
two sets $R_i$, $R_j$, such that both $R_i$ and $R_j$ are associated
with a \tyb solution in $\Pi''$ and $E_{i,j}$ is part of the instance $R$,
then we can improve the solution by imposing
a \tya solution for  $R_i$.
\qed

\vspace*{0.15cm}

A consequence of Lemmata~\ref{lem:opt-R_i} and~\ref{lem:compute-canonical}
and some properties of the instance $R$, is Lemma~\ref{lem:apx-hard1}.

\begin{Lemma}
\label{lem:apx-hard1}
Let $\Pi$ be a solution of $3$-AP($3$) over instance $R$
of cost $6|V|+3|C|+11|E|+9$, then
we can compute in polynomial time a solution of MVCC over instance $G$ of size $C$.
\end{Lemma}
\begin{proof}
Let us consider a canonical solution of $3$-AP($3$). 
First ,notice that the three rows $w_1$, $w_2$, $w_3$ provide together a cost of $9$.
Since two sets of rows are associated with a \tyb solution only if there does not
exist a set $E_{i,j}$, on the contrary, given an edge set $E_{i,j}$
at least one of the set $R_i$ and $R_j$ is associated with a \tya solution.
Consequently, the set of rows associated with a \tya solution corresponds
to a vertex cover of the graph $G$.

Now consider the cost of a canonical solution. For each set $R_i$ of rows associated
with a \tyb solution, we can show that each of the three edge sets
$E_{i,j}$, $E_{i,h}$, $E_{i,l}$ has a cost of $10$. Notice that, given 
an edge set $E_{i,j}$, 
if both sets $R_i$, $R_j$ are associated with \tya solutions, then we can show
that the edge set $E_{i,j}$ has a cost of $11$.
Accounting this decreasing of the cost of the edge sets to the set $R_i$ of
rows with a \tyb solution, is equivalent to assign to a \tyb solution a cost
equal to  $6$, while a \tya solution has a cost equal to $9$.
\end{proof}

Similarly to Lemma~\ref{lem:apx-hard1}, we can prove that
starting from a solution $C$ of MVCC over instance $G$, we can compute
in polynomial time a solution $\Pi$ of $3$-AP($3$) over instance $R$
of cost $6|V|+3|C|+11|E|+9$.
Therefore
$3$-AP($3$) is APX-hard.

\newpage

\section*{Appendix}

\subsection*{Proofs of Section~\ref{sec:par-hard}}

\subsubsection*{Proof of Lemma~\ref{lem:W[1]-hard:prop-base}}
\begin{Lemma}
\label{appendix-lem:W[1]-hard:prop-base}
Let $R$ be the instance of $\langle e \rangle$-AP associated with  $G$ and
consider two rows $r,r_x(i,j)$ of $R$, such that
$r \in R_0$ and $r_x(i,j) \in R(i,j)$. Then, $r[t] \neq r_x(i,j)[t]$,
for each $1 \leq t \leq 2h$.
\end{Lemma}
\begin{proof}
By construction, $r[t]=0$ for all $t$ with  $1 \leq t \leq 2h$,  while
$r_x(i,j)[t]=\sigma_{i,j}$.
\qed\end{proof}

\subsubsection*{Proof of Lemma~\ref{lem:W[1]-hard-1}}
\begin{Lemma}
\label{appendix-lem:W[1]-hard-1}
Let $G=(V,E)$ be a graph, let $V'$ be a  $h$-clique of $G$ and let
$R$ be the instance of $\langle e \rangle$-AP associated with  $G$.
Then we can compute in polynomial time a solution $\Pi$ of $\langle e \rangle$-AP over instance $R$
with cost at most $6h^3$.
\end{Lemma}
\begin{proof}
Initially let $\Pi'$ be a solution consisting of clusters $R_0$, $R(i,j)$, for each $R(i,j) \in R$.
For each $R(i,j$), let $r_1(i,j)$ be the first row of $R(i,j)$.
Compute a new solution
$\Pi$ consisting of clusters $R_0'$, $R'(i,j)$, for each $R(i,j) \in R$, where:
\begin{itemize}
\item $R'(i,j) = R(i,j) \setminus \{ r_1(i,j) \}$, for each $v_i, v_j \in V'$;
\item $R'(i,j) = R(i,j)$, for $v_i\notin V'$ or  $v_j \notin V'$;
\item $R_0'= R_0 \bigcup_{R(i,j) \in R} \left( R'(i,j) \setminus R(i,j) \right)$
\end{itemize}
Notice that, since $V'$ is a $h$-clique, $|R_0'|=k$.
Moreover, by construction, $|R(i,j)|\ge |R'(i,j)|\ge |R(i,j)|-1$, therefore
$\Pi$ is a feasible solution for $R$.
Notice also that no entries is suppressed in the rows of each set $R'(i,j)$, therefore to
determine the cost of $\Pi'$ it
suffices to determine the number of entries deleted in $R_0'$, and we will
show that such number is exactly  $6h^3$.

Indeed, by construction, for each  column $t$ of the  first $2h$ columns,
and for each row $r \in R_0$ and $r_x(i,j) \in R(i,j)$, $r[t] \neq r_x(i,j)[t]$,
hence all the entries
of the first $2h$ columns of the rows in $R_0'$ must be deleted, resulting in $2hk$ suppressions.
Now let us consider the columns with index $2h+1  \leq t \leq 2h+n$ and $v_t \in
V'$. In such
positions, all rows of $R_0$ are equal to $0$, while all rows in the sets $R(y,t)$, $R(t,y)$ are equal to $1$.
Consider the $\frac{h(h-1)}{2}$ of $R_0' \setminus R_0$. As the corresponding
edges are incident on a set of $h$ vertex, by construction
there exists a set $H$ of exactly $h$ columns,
with $H= \{t:2h+1 \leq t \leq 2h+n \}$,
where at least one of the rows in $R'_0 \setminus R_0$
is equal to $1$, while the rows in $R_0$ are all equal to $0$.
Since in any other column all  rows in $R_0'$ have value equal to $0$,
hence there are additional $hk$ suppressions for the columns
with index $2h+1 \leq t \leq 2h+n$.
Overall, the number of suppressions is $3hk$ which, by the choice of $k$ is
equal to $6h^3$.
\qed\end{proof}


\subsection*{Proofs of Section~\ref{sec:fpt-algo}}

\subsubsection*{Proof of Lemma~\ref{lem:feas:solution}}

\begin{Lemma}
\label{appendix-lem:feas:solution}
Let $M$ be a maximum weight matching of \grafoSR, then the solution $\PiS(M)$
computed by Algorithm~\ref{alg:matching} is
feasible.
\end{Lemma}

Lemma~\ref{lem:feas:solution} is a consequence of Lemmata~\ref{lem:feas-matching-feas-sol}  and~\ref{lem:match:opt:feas}.

\begin{Lemma}
\label{lem:match:opt:feas} Let $M$ be a maximum weight matching of
\grafoSR, then $M$ is a feasible matching.
\end{Lemma}
\begin{proof}
First notice that, as $M$ is feasible, each vertex of $T$ is covered
and each vertex $R_{safe}^l(r,j)$, with $1 \leq j \leq k$, is covered by $M$.
Assume that $M$ is not complete and that a vertex $R_{dist}^l(r)$ of $R^l_{dist}$ (resp.
$R_{safe}^l(r,k+j)$, with $1 \leq j \leq exc(g)$, of $R^l_{safe}$) is not matched. Then,
by construction, also the vertex $R_{dist}^r(r)$
of $R^r_{dist}$ (resp. $R_{safe}^r(r,j)$ of $R^r_{safe}$) is not covered by $M$,
as $R_{dist}^l(r)$ (resp. $R_{safe}^l(r,k+j)$) is the only vertex adjacent to
$R_{dist}^r(r)$ (resp. $R_{safe}^r(r,j)$) in \grafoSR.
Hence we can compute the matching $M'$ by adding all the edges of $M$ to $M'$ and
by adding  edges
$(R_{dist}^l(r) , R_{dist}^r(r))$, (resp. $(R_{safe}^l(r,k+j),R_{safe}^r(r,j))$)
for each vertex $R_{dist}^l(r)$ (resp. $R_{safe}^l(r,k+j)$)
not covered by $M$.
\qed\end{proof}

\begin{Lemma}
\label{lem:prop:match:part:feas}
Let $M$ be a feasible matching of \grafoSR,
if $M$ is not complete, then we can compute in polynomial time a complete matching $M'$,
such that $w_h(M') > w_h(M)$.
\end{Lemma}

As a consequence of Lemma~\ref{lem:prop:match:part:feas}, we assume in what
follows that any matching $M$ is complete. Furthermore,
we can prove the following result.

\begin{Lemma}
\label{lem:feas-matching-feas-sol}
Let $M$ be a complete matching of \grafoSR. Then
Algorithm~\ref{alg:matching}
computes in polynomial time a feasible clustering $\PiS(M)$.
\end{Lemma}
\begin{proof}
Since $\PiS(M)$  feasible,   all vertices in $T$ are covered by $M$.
Furthermore, we can assume, by Lemma~\ref{lem:prop:match:part:feas},
that each vertex in $R^l_{dist} \cup R^l_{safe}$ is covered by $M$.
Hence each row in $R$ is assigned by Algorithm~\ref{alg:matching}
to a set whose resolution vector is $S'$. Furthermore Algorithm~\ref{alg:matching}
assigns to each set with resolution vector $x \in S'$ at least $k$ rows. Hence the clustering
$\PiS(M)$ computed by Algorithm~\ref{alg:matching} is feasible.
\qed\end{proof}

\subsubsection*{Proof of Lemma~\ref{lem:cost:match:part:feas}}
\begin{Lemma}
\label{appendix-lem:cost:match:part:feas}
Let $M$ be a complete matching of \grafoSR,
then the total weight of $M$, $w_h(M)$, is equal to
%
$k \sum_{r \in S'} ( W + m -del(r)+1) + \sum_{r \in R(L)} (m - del(r)) = (W+1)
k |S'| + m |R^l_{dist} \cup R^l_{safe}|- (k\sum_{r \in S'} del(r) + \sum_{r \in R(L)} del(r))$.

\end{Lemma}
\begin{proof}
The total weight $w_h(M)$ of the matching $M$
is defined as \[
w_h(M) =  \sum_{(u,v) \in M(T)} w_h((u,v)) + \sum_{(u,v) \in M(L)} w_h(u,v).\]
By Lemma~\ref{lem:eq-matching} and by definition of the weight function $w_h$,
it follows that
\[
w_h(M) = k\sum_{r \in S'} w'(r) + \sum_{r \in R(L)} (m - del(r))
\]
and by definition of $w'(r)$ it holds
\[w_h(M) = k \sum_{r \in S'} ( W + m -del(r)+1) + \sum_{r \in R(L)} (m - del(r)).
\]
Hence
\[
w_h(M) =  (W+m+1) k |S'| - k \sum_{r \in S'} del(r) + \sum_{r \in R(L)} m - \sum_{r \in R(L)} del(r).
\]
By definition of feasible matching and by
Lemma~\ref{lem:prop:match:part:feas},
$|V(T)|=|T|$.
Furthermore, since $|T| = k |S'|$,
then $mk|S'| = m |T|= m |V(T)|$. By construction $\sum_{r \in R(L)} m = m|V(L)| $
and $V(T) \cup V(L) = R^l_{dist} \cup R^l_{safe}$.
Hence
\[w_h(M) = (W+1) k |S'| + m |R^l_{dist} \cup R^l_{safe}|- (k\sum_{r \in S'} del(r)
+ \sum_{r \in R(L)} del(r)).
\]
\qed\end{proof}

\subsubsection*{Proof of Lemma~\ref{lem:from-partially-feasible-solution-to-matching}}

\begin{Lemma}
\label{appendix-lem:from-partially-feasible-solution-to-matching}
Let $R$ be an instance of \smAP{}, let $\PiS$ be a feasible
solution of \smAP{} over instance $R$ that suppresses at most $e$ entries,
let \grafoSR be the graph associated with $R$ and $S'$.
Then there exists a complete  matching of \grafoSR with total weight
$w_G(M) \geq (W+1) k |S'| + m |R^l_{dist} \cup R^l_{safe} \cup R'^l_{safe}|- e$.
\end{Lemma}
\begin{proof}
Since $\PiS$ is feasible, we notice that
each set of $\PiS$ associated with a resolution vector $r \in S'$ must have cardinality at least $k$.
Furthermore, we assume that all the sets of $\PiS$ are all associated with different resolution vectors,
otherwise we can merge all the sets with the same resolution vector without increasing
the cost of $\PiS$.

Let $x$ be a row of $S'$ and denote by $R_x \subseteq R$ the set of rows of $R$
assigned to the set associated with resolution vector $x$.
Starting from $\PiS$ we compute incrementally a matching $M$ by adding edges.
First, for each set of vertices $T(x,i)$, $1 \leq i \leq k$,
let $i^*$ be the minimum number
such that $T(x,i^*)$ does not have any edge incident on it in $M$.
First, assume that $x \in S'_{safe}$;
add the edge $(R_{safe}'^l(g,i), T(x,i))$ to $M$, for each $1 \leq i \leq k$.
Now, assume that $x \in S'_{cost}$.
Scan the rows in $R_x$ and for each row $r$ in $R_x$,
if $r \in R_{dist}$ add the
edge $(R_{dist}^l(r), T(x,i^*))$ to $M$. If $r \in R_{safe}$ and belongs to group $g$
add the edge $(R_{safe}^l(g,i), T(x,i^*))$ to $M$.
If no such $T(x,i^*)$ exists, then no edge is added to $M$.
Notice that by construction, since all sets in $S'$ have at least $k$ rows, then all
vertices of $T$ are covered by $M$, therefore $M$ is feasible.

Finally add to $M$ all edges $(\Rdist^l(x),\Rdist^r(x))$,
$(\Rsafe^l(g,i),\Rsafe^r(g,i))$, $1 \leq i \leq exc(g)$,
for each vertex in $\{ \Rdist^l$, $\Rsafe^l \}$ respectively
that is not already covered in $M$. 
Hence  $M$ is complete.

Given a solution $\PiS$, a resolution vector $x$ of $S'$ and the corresponding matching $M$, consider the order in which the rows of a set $R_x$
are scanned sequentially to construct $M$.
Each of the first $k$ rows assigned to a cluster with resolution vector equal
to $x$, by construction corresponds to an edge of $M$ joining
a vertex of $V(T)$ and a vertex of
$T$. Since $M$ is complete, those rows have a total cost in $\PiS$ of $k\sum_{r \in S'} del(r)$.
The remaining rows of $R$ correspond to vertices of
$V(L)$.
Notice that those rows have a total cost in $\PiS$ not larger than
$\sum_{r \in R(L)} del(r)$.
By Lemma~\ref{lem:cost:match:part:feas}
$w_h(M) = (W+1) k |S'| + m |R^l_{dist} \cup R^l_{safe} \cup R'^l_{safe}|- (k\sum_{r \in S'} del(r)+
\sum_{r \in R(L)} del(r))$. Since $\PiS$ suppresses at most $e$ entries of $R$,
then
$e  \leq k\sum_{r \in S'} del(r)+ \sum_{r \in R(L)} del(r)$, therefore
$w_h(M) \geq (W+1) k |S'| + m |R^l_{dist} \cup R^l_{safe} \cup R'^l_{safe}|- e$.
\qed\end{proof}

\subsection*{Proofs of Section \ref{sec:AP-3-3}}

It is easy to see that, by construction, the following properties hold.

\begin{Proposition}
\label{prop:dist-vertex-rows}
Let $r_a$, $r_b$ be two rows of $R_i$, with $r_a=g_j(v_i)$ and
$r_b=g_l(v_i)$, $j < l$. Let $r_c$ be a row of $R_j$, with
$i \neq j$. Then:

\begin{itemize}

\item $H(r_a, r_c)=H(r_b, r_c)=3$;

\item $H(r_a,r_b) \leq 2$;

\item $H(r_a,r_b) = 1$ iff $r_a=g_h(v_i)$ and $r_b=g_{h+3}(v_i)$, with
$1 \leq h \leq 3$, or $r_a=g_h(v_i)$ and $r_b=g_l(v_i)$, with
$4 \leq j \leq l \leq 6$.

\end{itemize}

\end{Proposition}

\begin{Proposition}
\label{prop:dist-edge-rows}
Let $r_a$, $r_b$ be two rows of $E_{i,j}$, with $r_a \in g_h(v_i,v_j)$ and
$r_b \in g_l(v_i,v_j)$, with $h < l$. Let $r_c$, $r_d$ be two rows of $R_i$ and $R_p$,
with $p \neq i,j$, and let
$r_e$ be a row of $E_{t,z}$,  with $t \neq i$ or $z \neq j$. Then:

\begin{itemize}

\item $H(r_a, r_b) \leq 2$;

\item $H(r_a,r_b) = 1$ iff $r_a=g_h(v_i,v_j)$ and $r_b=g_{h+1}(v_i,v_j)$, with
$1 \leq h \leq 2$;

\item $H(r_a, r_c) = 1$ iff $r_c$ is in the docking group $d_{i,j}(g(v_i))$ of $R_i$
and $r_a$ is in the group $g^i(v_i,v_j)$;

\item $H(r_a, r_c) = 2$ only if $r_c$ is in a group adjacent to $d_{i,j}(g(v_i))$;

\item $H(r_a, r_d) = 3$;

\item $H(r_a, r_e) = 3$.

\end{itemize}

\end{Proposition}

In what follows, by an abuse of notation, we may use a group $g(\cdot{})$ to denote its representative row $r_g(\cdot{})$.
Fig.~\ref{fig:graph-dist-1} shows the groups of $R_i$, $R_j$, $E_{i,j}$. Each group
of identical rows si represented with a vertex, while an edge
joins two vertices iff the corresponding groups are at Hamming distance $1$.

\begin{figure}[htb!]
\begin{center}
\includegraphics[width=12cm]{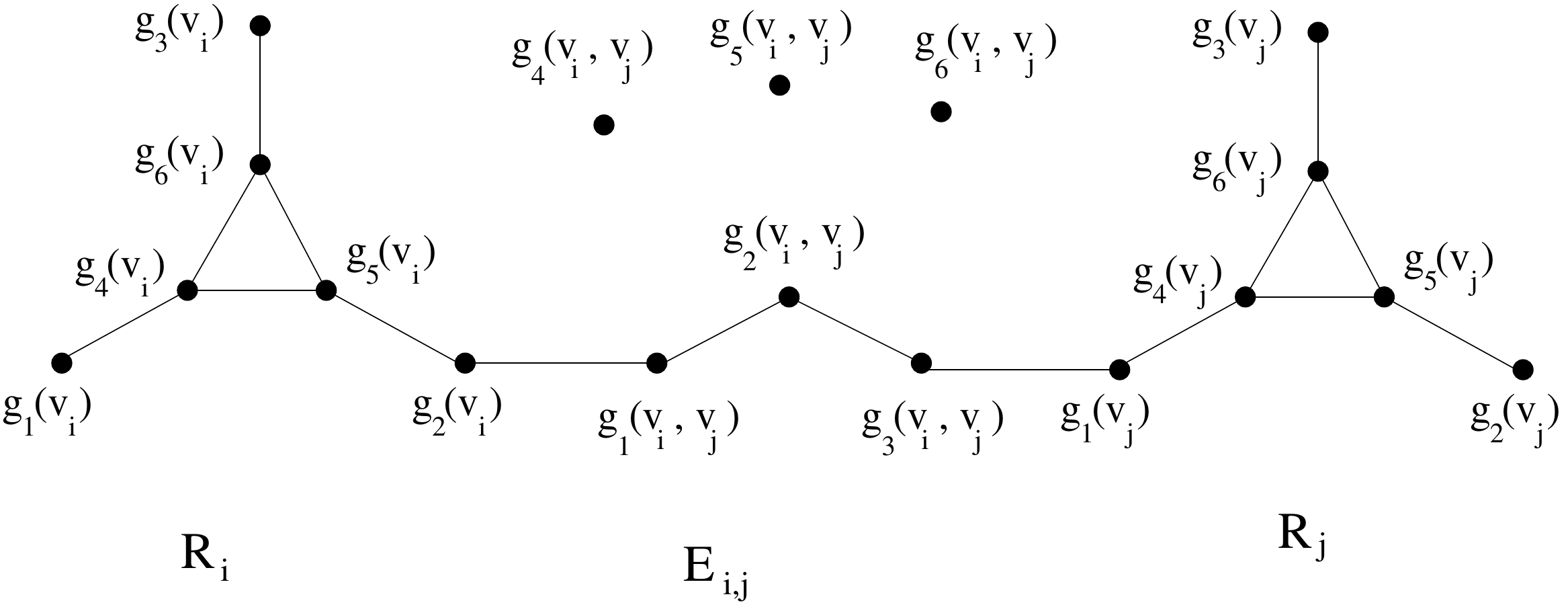}
\label{fig:graph-dist-1}
\caption{Groups at Hamming distance $1$ in $R_i$, $R_j$, $E_{i,j}$:
  vertices represent groups, while an edge joins two vertices
  representing groups at Hamming distance $1$.}
\end{center}
\end{figure}

\subsubsection*{Proof of Lemma \ref{lem:opt-R_i}}
\begin{Lemma}
\label{appendix-lem:opt-R_i}
Let $R_i$ be a set of rows, then an optimal solution of
$3$-AP($3$) over instance $R_i$ has a cost of at least $9$.
\end{Lemma}
\begin{proof}
Let us consider the set of $9$ rows, distributed in $6$ groups, of $R_i$.
As none of the group of $R_i$ consists of at least $3$ rows, it follows that any solution of $3$-AP($3$)
suppresses at least one entry in each row of $R_i$, hence
the lemma follows.
\end{proof}

\subsubsection*{Proof of Lemma \ref{lem:compute-canonical}}

In order to prove Lemma \ref{lem:compute-canonical}, first we have to show
some properties of a canonical solution.

\begin{Lemma}
\label{lem:opt-ty_a}
Let $R_i$ be a set of rows. Then  a solution of $3$-AP($3$) over instance $R$ induces an optimal cost
for the set $R_i$ if it is a \tya solution.
\end{Lemma}
\begin{proof}
By construction, a \tya solution is an optimal solution over instance $R_i$, as each
row has a cost of $1$ in a \tya solution.
\end{proof}

%
Now, in Lemma~\ref{lem:cost-ty_b}, we will prove a property
of a \tyb solution over the sets $R_i$, $E_{i,j}$, $E_{i,h}$, $E_{i,k}$. 

\begin{Lemma}
\label{lem:cost-ty_b}
Let $S$ be a \tyb solution of $3$-AP($3$) over instance
$R_i \cup E_{i,j} \cup E_{i,h} \cup E_{i,k}$, then $S$ suppresses $9$ entries in the rows of $R_i$,
and $1$ entry for each row of $g^i(v_i,v_j)$, $g^i(v_i,v_h)$, $g^i(v_i,v_k)$.
\end{Lemma}
\begin{proof}
By construction each set of a \tyb solution
containing a docking group of $R_i$ consists of three rows, where exactly one position is suppressed
for each row by Prop.~\ref{prop:dist-edge-rows}.
The cluster $g_4(v_i) \cup g_5(v_i) \cup g_6(v_i)$ consists of $3$ rows, where exactly one position
is suppressed for each row.
By a simple counting argument, the rows of $R_i$ have a total cost of $9$.
\end{proof}

Let $\Pi$ be a solution of $3$-AP($3$) over instance $R$ and 
let $E_{i,j}$ be an edge set.
Then we say that $\Pi$ induce an \emph{i-normal} solution for $E_{i,j}$
if it contains the following three sets:
(i) one set clusters  $C_1 = g^i(v_i,v_j) \cup d_{i,j}(v_i)$, 
(ii) one set containing $C_2 \supseteq \bigcup_{t \in \{ 4,5,6 \}} g_t(v_i,v_j)$, such that  
exactly two entries (in columns $1$ and $3$) are suppressed in the rows of $C_2$,
(iii) set $C_3 = g_2(v_i,v_j) \cup g^j(v_i,v_j)$.

\begin{Lemma}
\label{lem:cost-edge-10}
Let $\Pi$ be a solution of $3$-AP($3$) over instance $R$ and 
let $E_{i,j}$ be an edge set, 
then $\Pi$  suppresses at most $10$ entries in the rows of $E_{i,j}$
only if it induces an 
i-normal or j-normal solution for $E_{i,j}$.
\end{Lemma}
\begin{proof}
First assume that $\Pi$ induces an i-normal solution for $E_{i,j}$.
By Prop. \ref{prop:dist-edge-rows}, it follows that one entry 
of $g^i(v_i,v_j)$ is suppressed.
In the set consisting of the rows in $g_t(v_i,v_j)$, with $t \in \{4,5,6 \}$, 
by Prop. \ref{prop:dist-edge-rows} two positions for each row are suppressed.
Finally, by Prop. \ref{prop:dist-edge-rows}, in the set $g_2(v_i,v_j) \cup g^j(v_i,v_j)$,
exactly one position is suppressed for each row.

Now, let us prove that if $\Pi$ is a solution that is not i-normal or j-normal
for $E_{i,j}$, then $c_{\Pi}(E_{i,j}) \geq 11$. Notice that
that each row in $g_t(v_i,v_j)$, with $t \in \{4,5,6 \}$, has a Hamming
distance $2$ from any other row of $R \setminus g_t(v_i,v_j)$, hence
at least two entries are suppressed in each solution $\Pi$.
Furthermore, notice that each of the 
four rows in the groups $g_1(v_i,v_j)$, $g_2(v_i,v_j)$, $g_3(v_i,v_j)$ must have a cost of at most $1$.
But then, the rows of $g_2(v_i,v_j)$ must be co-clustered with the row of exactly one of
$g_1(v_i,v_j)$, $g_3(v_i,v_j)$ (w.l.o.g. $g_1(v_i,v_j)$). But then $g_3(v_i,v_j)=g^i(v_i,v_j)$,
must be co-clustered with $d_{i,j}(v_i)$.
\end{proof}

\begin{Lemma}
\label{lem:set-X}
Let $S$ be a solution of $3$-AP($3$) over instance $R$,
then we can compute in polynomial time a solution $S'$ such that
$c(S') \leq c(S)$ and $S'$ contains at most one set suppressing
three entries for each row.
\end{Lemma}
\begin{proof}
Assume that solution $S$ contains sets $Y_1, \dots, Y_p$, with $ p \geq 2$, such that
all the positions of the rows in $Y_j$ are suppressed. Then we can compute
in polynomial time a solution $S'$ by merging the set $Y_1, \dots, Y_p$ in a single
cluster $Y$. Notice that $c(S') \leq c(S)$, as in both solution $S'$ and $S$
three positions are suppressed for each row $r \in \bigcup_{j=1 \dots p}Y_j$.
\end{proof}

Now, let us first introduce some properties of a solution of $3$-AP($3$)
over instance $R$.

\begin{Lemma}
\label{lem:cost-edge-mand-set}
Let $\Pi$ be a solution of $3$-AP($3$) over instance $R$, we can compute in polynomial time
a solution $\Pi'$ such that
\begin{enumerate}
\item for each edge set $E_{i,j}$, $\Pi'$ has a set $S_{i,j}$ containing the rows of groups $\bigcup_{t = 4,5,6} g_t(v_i,v_z)$, such that for each row in $S_{i,j}$ exactly two columns (columns $1$ and $3$) are supprssed;
\item $c(\Pi') \leq c(\Pi)$.
\end{enumerate}
\end{Lemma}
\begin{proof}
First, notice that by Prop.~\ref{prop:dist-edge-rows} the rows in the groups
$\cup_{t = 4,5,6} g_t(v_i,v_j)$ have distance smaller than $3$ only
w.r.t. rows of $E_{i,j}$.
Furthermore, notice that, by construction, each row in $\bigcup_{t = 4,5,6} g_t(v_i,v_j)$
may be equal to another row of $R$ only in the second position.

Assume that there exist clusters $S_1, S_2, S_3$ (at most one of these clusters can be empty)
containing the rows of $\bigcup_{t = 4,5,6} g_t(v_i,v_j)$,
such that at most two entries are suppressed for each row of the cluster
$S_j$, $j \in \{ 1,2,3 \}$. Then, for each row in $S_1 \cup S_2 \cup S_3$,
the positions $1$ and $3$ are suppressed.
Hence, we can merge clusters $S_1, S_2, S_3$, without increasing the cost of the solution,
obtaining one set that contains the rows $\bigcup_{t = 4,5,6} g_t(v_i,v_z)$.

Assume that some rows of $\cup_{t = 4,5,6} g_t(v_i,v_j)$ are in the cluster $X$
and some rows of $\cup_{t = 4,5,6} g_t(v_i,v_j)$ are in  a different cluster $Y$,
such that at most two entries are suppressed for each row of the cluster
$Y$.
Then we can move the rows of $\cup_{t = 4,5,6} g_t(v_i,v_j) \cap X$ to $Y$, decreasing the
cost of the solution.

Now, assume that these rows are all clustered in set $X$. It follows that each row
of $\cup_{t = 4,5,6} g_t(v_i,v_j)$ have a cost of $3$. Hence, we can move
this set of rows to a new set $\cup_{t = 4,5,6} g_t(v_i,v_j)$,
decreasing the cost of the solution.
\end{proof}

\begin{Lemma}
\label{lem:cost-vert-adj-opt}
Let $R_i$, $R_j$ be two set of rows and let $E_{i,j}$ be an edge set of $R$.
Let $\Pi$ be a solution of $3$-AP($3$) over instance $R$ that associates a \tyb solution
with both $R_i$, $R_j$.
Then we can compute in polynomial time a solution $\Pi'$ of $3$-AP($3$) over instance $R$
where exactly one of $R_i$, $R_j$ is associated with a \tyb solution and
such that $c(\Pi') \leq c(\Pi)$.
\end{Lemma}
\begin{proof}
Notice that, by Prop.~\ref{prop:dist-edge-rows}, the rows of group $g_2(v_i,v_j)$
have Hamming distance $1$ only from the rows of $g_1(v_i,v_j)$ and
$g_3(v_i,v_j)$. Since in a \tyb solution the rows of $g_1(v_i,v_j)$ and
$g_3(v_i,v_j)$ are co-clustered with rows of $R_i$ and $R_j$, it follows by Prop.~\ref{prop:dist-edge-rows}
that the rows of $g_2(v_i,v_j)$ are co-clustered in $\Pi$ with rows at Hamming distance
at least $2$. Hence, $\Pi$ suppresses two entries in each row of $g_2(v_i,v_j)$,
and, as $g_2(v_i,v_j)$ consists of $2$ rows, at least $4$ entries of rows
in $g_2(v_i,v_j)$ are suppressed in $\Pi$.
Notice that, as $R_i$ and $R_j$ are associated with \tyb solutions in $\Pi$, the only rows that can be clustered
with the rows of $g_2(v_i,v_j)$ are those of groups  $g_4(v_i,v_j)$, $g_5(v_i,v_j)$, $g_6(v_i,v_j)$.

Starting from solution $\Pi$,
let us compute a solution $\Pi'$ of $3$-AP($3$) over instance $R$ as follows.
Let $E_{i,j}$, $E_{i,t}$, $E_{i,z}$ be the edge sets associated with the three edges
incident in $v_i$.
Modify solution $\Pi$ so that $\Pi'$ induces a \tya solution for $R_i$, and
a j-normal solution for $E_{i,j}$.
Moreover, for each row of a group
$g^i(v_i,v_z)$ of an edge set $E_{i,z}$, with $z \neq j$,
co-cluster such group $g^i(v_i, v_z)$ with the cluster containing the rows of
$\bigcup_{t = 4,5,6} g_t(v_i,v_z)$.

By Lemma \ref{lem:cost-edge-10} the rows of edge set $E_{i,j}$ have a total cost of $10$.
Notice that by Lemma~\ref{lem:cost-edge-mand-set}, we can assume that $\Pi$ has
a set $C$ containing $\bigcup_{t = 4,5,6} g_t(v_i,v_z)$, such that exactly
two entries (corresponding to the positions $1$ and $3$)
are suppressed for each row in $C$. Hence the representative row of $C$
has Hamming distance $2$ from $g^i(v_i, v_z)$, as they are equal in position $2$.

Now, each of these two rows $g^i(v_i, v_z)$ has a cost of $2$ in $\Pi'$, while it has a cost of least $1$
in $\Pi$. Notice that each of the two rows of $g_2(v_i,v_j)$ has a cost of at least $2$ in $\Pi$,
while it has a cost of $1$ in $\Pi'$.
Hence, $c(\Pi')\leq c(\Pi)$.
\end{proof}

\begin{Lemma}
\label{lem:APX-edge-set-cost}
Let $\Pi$  be a solution of the $3$-AP($3$) over instance $R$, such that two sets
$R_i$, $R_j$ are not associated with a \tyb solution in $\Pi$ and
$\Pi$ induces a total cost of $10$ for the rows in $E_{i,j}$.
Then at least one of $R_i$, $R_j$ has cost $11$.
\end{Lemma}
\begin{proof}
Assume that $\Pi$ induces a total cost of $10$ for the rows in $E_{i,j}$.
Notice that the rows in $\bigcup_{t=4,5,6}g_t(v_i,v_j)$
have a total cost of $6$, as
by Prop.~\ref{prop:dist-edge-rows} they are at Hamming distance at least $2$ from any other row
of $R$. Furthermore, the $4$ rows of $E_{i,j} \setminus \bigcup_{t=4,5,6}g_t(v_i,v_j)$
must have a cost of at least $1$ in $\Pi$.
Notice that, by Lemma \ref{lem:cost-edge-10},
 $\Pi$ induces either an i-normal or j-normal solution for $E_{i,j}$ (w.l.o.g.
we assume that is i-normal). Hence
$g^i(v_i,v_j)$ is clustered
with the rows of $d_{i,j}(v_i)$), while of $g^j(v_i,v_j)$ 
is clustered with $g_2(v_i,v_j)$, otherwise some
rows of $g_2(v_i,v_j)$ are clustered in $\Pi$ with a row at Hamming distance
at least $2$, hence the total cost of the rows in $E_{i,j}$ is greater than $10$.
But then, we claim that $\Pi$ induces a cost of at least
$11$ for the rows of the set $R_i$.

Now, recall that by hypothesis
$R_i$ is not associated with a \tyb solution in $\Pi$,
and let us consider the clusters containing rows of $R_i$ in $\Pi$. Notice that
if at least two rows of $R_i$ are clustered with some rows
at Hamming distance at least $2$, then $\Pi$ induces a cost of at least $11$ for the set $R_i$.
Recall that group $d_{i,j}(g(v_i))$ of $R_i$ is clustered
only with rows of group $g^i(v_i,v_j)$ of $E_{i,j}$, and consider the cases that either the three
groups of rows in $g_4(v_i), g_5(v_i), g_6(v_i)$ are co-clustered, or not.
In the former case, as
$R_i$ is not associated with a \tyb solution, it follows that the rows of at least one
of the docking group of $R_i$ are clustered with rows at Hamming distance $2$;
hence $\Pi$ induces a cost of at least $11$ for the set $R_i$.
In the latter case, let us consider the group of $R_i$ (w.l.o.g. $g_4(v_i)$) adjacent to
$d_{i,j}(g(v_i))$ and let $C$ be the cluster
containing the unique row of $g_4(v_i)$. As the rows in $g_4(v_i), g_5(v_i), g_6(v_i)$ are
not co-clustered, it follows that $C$ contains a row $r$ at Hamming distance
at least $2$ from $g_4(v_i)$. If $r \in R_i$, then
$\Pi$ suppresses at least two entries of two rows of $R_i$,
namely $r$ and $g_4(v_i)$,
hence
$\Pi$ induces a cost of at least $11$ in rows of $R_i$.
If $r \notin R_i$, then $r$ must be a row at Hamming distance
$3$ from $g_4(v_i)$. Indeed by Prop.~\ref{prop:dist-vertex-rows} and by Prop.~\ref{prop:dist-edge-rows},
the rows at Hamming distance
not greater than $2$ from $g_4(v_i)$ belong to $R_i \cup (E_{i,j} \setminus \bigcup_{t=4,5,6}g_t(v_i,v_j))$.
We have assumed that $(R_i \setminus \{ g_4(v_i) \}) \cap C = \emptyset$,
and it must be $( E_{i,j}\setminus \bigcup_{t=4,5,6}g_t(v_i,v_j)) \cap C = \emptyset$,
since by hypothesis $\Pi$ induces an i-normal solution for the rows in $E_{i,j}$.
Hence $g_4(v_i)$  must have
cost equal to $3$ and must be part of the cluster $X$ in $\Pi$ by the Lemma~\ref{lem:set-X}.
It follows that $\Pi$ induces a cost of at least
$11$ for the set $R_i$.
\end{proof}

Now, let us prove Lemma \ref{lem:compute-canonical}.
\begin{Lemma}
\label{appendix-lem:compute-canonical}
Let $\Pi$  be a solution of $3$-AP($3$) over instance $R$. Then we can compute in polynomial time a canonical solution $\Pi'$ of $3$-AP($3$) over instance $R$ such that $c(\Pi') \leq c(\Pi)$.
\end{Lemma}
\begin{proof}
Let us consider the solution $\Pi$.
Before computing a canonical solution $\Pi'$, we
compute an intermediate solution $\Pi''$ such that
$c(\Pi'') \leq c(\Pi)$ as follows.
First, for each set $R_i$, if $R_i$ is associated with
a $\tyb$ solution in $\Pi$, then define a $\tyb$ solution for $R_i$ in $\Pi''$.
Otherwise, if each docking vertices $d_{i,j}(v_i)$ of a set $R_i$ is clustered in $\Pi$ with
the row of group $g^i(v_i,v_j)$ of $E_{i,j}$, then define a \tyb solution for $R_i$ in $\Pi''$;
else define a \tya solution for $R_i$ in $\Pi''$. Furthermore,
define a set containing row $x_1, x_2, x_3$ in $\Pi''$.
Next, consider the rows of an edge set $E_{i,j}$, and
define a clustering of the rows not yet clustered in $\Pi''$.
If exactly one of $R_i$, $R_j$ (w.l.o.g. $R_i$) is associated
with a \tyb solution in $\Pi''$, then define an i-normal solution for $E_{i,j}$ in $\Pi''$.
Else if at least one of $R_i$, $R_j$
(w.l.o.g. $R_i$) is associated
with a \tya solution in $\Pi''$, then define
the following solution: one set contains the rows in $g^i(v_i,v_j) \cup g_2(v_i,v_j)$;
one set contains the rows in  $\bigcup_{t=4,5,6} (g_t(v_i,v_j)) \cup g^j(v_i,v_j)$.
If both $R_i$, $R_j$ are associated with a \tyb solution in $\Pi''$, then 
define a set
$\bigcup_{t=4,5,6} (g_t(v_i,v_j)) \cup g_2(v_i,v_j) $ in $\Pi''$.
%

Now, let us show that $c(\Pi) \geq c(\Pi'')$.
By Lemma~\ref{lem:opt-ty_a} and by Lemma \ref{lem:cost-ty_b}
it follows that for each row in a set $R_i$
the cost in $\Pi''$ is optimal.
Furthermore, by Lemma~\ref{lem:set-X}, we can assume that $\Pi$ contains
a set $X \supseteq \{ x_1, x_2, x_3 \}$,  hence the rows in $\{ x_1, x_2, x_3 \}$
have all cost $3$ in both $\Pi$ and $\Pi''$.
Hence it remains to consider the cost of the edge set $E_{i,j}$.

Let $E_{i,j}$ be an edge set. Notice that by Lemma~\ref{lem:cost-edge-mand-set}
we can assume that $\Pi$ contains a set $S_{i,j} \supseteq \bigcup_{t=4,5,6}g_t(v_i,v_j)$,
and by construction $\Pi''$ contains a set $S'_{i,j} \supseteq \bigcup_{t=4,5,6}g_t(v_i,v_j)$.
Hence each row of $g_t(v_i,v_j)$, with $t \in \{ 4,5,6 \}$, has a cost
equal to $2$ in both $\Pi$, $\Pi''$.
Let us consider the case when both sets
$R_i$ and $R_j$ are associated with a \tyb solution in both $\Pi$ and $\Pi''$.
The groups of $E_{i,j}$ not co-clustered in a \tyb solution of
$R_i$, $R_j$, are $g_2(v_i,v_j)$, $g_4(v_i,v_j)$,
$g_5(v_i,v_j)$, $g_6(v_i,v_j)$. By construction,
as the rows at Hamming distance $1$ from $g_2(v_i,v_j)$ are clustered in the
\tyb solution of $R_i$, $R_j$ in $\Pi$ (hence cannot be co-clustered with $g_2(v_i,v_j)$),
it follows that
the rows $g_2(v_i,v_j)$ must be clustered with a row having Hamming distance
at least $2$ in $\Pi$.
As $\Pi''$ contains the set $S''_{i,j} = (\bigcup_{t=4,5,6}g_t(v_i,v_j)) \cup g_2(v_i,v_j)$
and as $\Pi$ contains the set $S_{i,j} \supseteq \bigcup_{t=4,5,6}g_t(v_i,v_j)$,
it follows that the cost of the rows in $E_{i,j}$ in solution $\Pi$ is greater or equal than the cost of the rows
in $E_{i,j}$ in solution $\Pi''$.

Let us consider the case when exactly one of the sets
$R_i$ and $R_j$ (w.l.o.g. $R_j$) is  associated with a \tyb solution in $\Pi''$ and in $\Pi$.
By construction, $R_i$ is associated with a
\tya solution in $\Pi''$. By Lemma~\ref{lem:cost-edge-10} it follows that that $c_{\Pi''}(E_{i,j}) = 10$,
and, as each row in $\bigcup_{t=4,5,6} g_t(v_i,v_j)$ has a cost of $2$ in $\Pi''$,
it follows that each row in $E_{i,j} \setminus (\bigcup_{t=4,5,6} g_t(v_i,v_j)) $ has a cost of $1$ in $\Pi''$.
As $\Pi$ contains the set $S_{i,j}$, it follows
that $\Pi$ suppresses two entries in the rows of $\bigcup_{t=4,5,6} g_t(v_i,v_j)$,
hence the cost of the rows in $E_{i,j}$ in solution $\Pi$ is greater or equal than the cost of
the rows in $E_{i,j}$ in solution $\Pi''$.

Let us consider the case when at least one of the sets
$R_i$ and $R_j$ (w.l.o.g. $R_i$) is associated with a \tyb solution in $\Pi''$ and not in $\Pi$.
Notice that by construction, the rows in groups $d_{i,j}(v_i)$, $g^i(v_j,v_j)$
are clustered in both $\Pi$ and $\Pi''$. Now, if $\Pi$ induces a cost of at least $11$
for the rows in $E_{i,j}$, since $\Pi''$ induces a cost of at most $11$
for the rows in $E_{i,j}$ it follows that
$c_{\Pi}(E_{i,j}) \geq c_{\Pi''}(E_{i,j})$.
If $\Pi$ induces a cost of $10$
for the rows in $E_{i,j}$, then by Prop. \ref{prop:dist-edge-rows} $g_2(v_i,v_j)$ must be co-clustered with
$g^j(v_i,v_j)$. Then, it follows that by construction $R_j$ is associated
with a \tya solution in $\Pi''$ and that the rows of $E_{i,j}$ have a total
cost of $10$ in $\Pi''$. Hence $c_{\Pi}(E_{i,j}) \geq c_{\Pi''}(E_{i,j})$.

Now, let us consider the case when both $R_i$, $R_j$, are associated with a \tya
solution in $\Pi''$. In this case, by construction, the rows in the edge set $E_{i,j}$
have a total cost of $11$ in $\Pi''$, while they have a cost of at least $10$ in $\Pi$,
as the rows in $\bigcup_{t=4,5,6}g_t(v_i,v_j)$ (contained in the set $S_{i,j}$ of $\Pi$)
have a total cost of $6$, while each of the $4$ rows of $E_{i,j} \setminus \bigcup_{t=4,5,6}g_t(v_i,v_j)$
has a cost of at least $1$ in $\Pi$.
Assume that the rows of $E_{i,j}$ have a total cost of $10$ in $\Pi$.
By Lemma~\ref{lem:APX-edge-set-cost}, $\Pi$ induces a total cost of $11$ for the
rows of one of the sets $R_i$, $R_j$ (w.l.o.g. $R_i$).
Notice that by Lemma~\ref{lem:opt-ty_a} $\Pi''$ induces a cost of $9$ for all the set $R_i$.
Now, let us consider the set $R_i$ and the three edge sets
$E_{i,j}$, $E_{i,h}$, $E_{i,k}$.
In what follows, we will consider the cost induced by $\Pi$ and by $\Pi''$
in the set $R_i$ and in some of the edge sets $E_{i,j}$, $E_{i,h}$, $E_{i,k}$.
More precisely, for each edge set $E_{i,x}$ in $\{ E_{i,j}, E_{i,h}, E_{i,k} \}$, let us consider
its cost together with the cost of $R_i$ only if
$d_{i,x}(v_i)$ and $g^i(v_i,v_x)$ are clustered in $\Pi$ (otherwise
$E_{i,x}$ will be eventually be considered together with $R_x$).
By construction the cost of at most two edge sets in $E_{i,j}$, $E_{i,h}$, $E_{i,k}$
(assume w.l.o.g. $E_{i,j}$, $E_{i,h}$) are considered together with the cost of $R_i$,
otherwise $d_{i,x}(v_i)$ and $g^i(v_i,v_x)$ would be co-custered in $\Pi$, for each $x \in \{ j,h,k \}$
and by construction $R_i$ would be associated with a \tyb solution in $\Pi''$.
Since
$c_{\Pi}(E_{i,j})\geq 10$, $c_{\Pi}(E_{i,h})\geq 10$, $c_{\Pi}(R_i)\geq 11$,
while $c_{\Pi''}(E_{i,j}) = 11$, $c_{\Pi''}(E_{i,h}) = 11$ and $c_{\Pi}(R_i)= 9$
it follows that
$c_{\Pi}(E_{i,j}) + c_{\Pi}(E_{i,h}) + c_{\Pi}(R_i)
\geq c_{\Pi''}(E_{i,j}) + c_{\Pi''}(E_{i,h}) + c_{\Pi''}(R_i)$.


Now, we have shown that $c(\Pi) \geq c(\Pi'')$. Notice that $\Pi''$ may not
be a canonical solution, as there may exist two sets $R_i$, $R_j$,
with $E_{i,j}$ part of the instance, associated with a \tyb solution in $\Pi''$.
Now, applying Lemma~\ref{lem:cost-vert-adj-opt} for each pair of sets $R_i$, $R_j$,
associated with a \tyb solution in $\Pi''$, with $E_{i,j}$ part of the instance,
we can compute a canonical solution
$\Pi'$ such that $c(\Pi'') \geq c(\Pi')$. Hence $c(\Pi) \geq c(\Pi')$.
\end{proof}


\vspace*{0.15cm}

\begin{Lemma}
\label{lem:apx-hard2}
Let $C$ be cover of $G$. Then, we can compute in polynomial time
a solution $\Pi$ of $3$-AP($3$) over instance $R$
of cost $6|V|+3|C|+11|E|+9$.
\end{Lemma}
\begin{proof}
We can define a solution $\Pi$ of $3$-AP($3$) of cost
$6|V|+3|C|+11|E|+9$, as follows. Define a \tya solution for each $R_i$ associated
with a vertex $v_i \in C$. Each of such sets has a cost of $9$.

Define a \tyb solution for set $R_i$ associated with a vertex $v_i \in V-C$, 
and define an i-normal solution for the sets $E_{i,j}$, $E_{i,h}$, $E_{i,l}$.
Each such set $R_i$ has a cost of $9$, and each edge set in $\{E_{i,j}$, $E_{i,h}$, $E_{i,l} \}$
has a cost of $10$.
Accounting this decreasing of the cost of the edge sets (from $11$ to $10$) to the set $R_i$, 
is equivalent to assign to a \tyb solution a cost equal to  $6$.

For any other edge set $E_{i,j}$ add to $\Pi$ the following sets:
$S_1 = g_1(v_i,v_j) \cup g_2(v_i,v_j)$, $S_2 = g_3(v_i,v_j) \cup g_4(v_i,v_j) \cup g_5(v_i,v_j)
\cup g_6(v_i,v_j)$. Each such edge set has a cost of $11$.
Finally, define a set $X= \{x_1, x_2, x_3 \}$, having a total cost of $9$.
\end{proof}

\end{document}